\documentclass[11pt,conference,letterpaper,onecolumn]{article}
\usepackage{fullpage}
\usepackage{amsthm}
\usepackage{amssymb}
\usepackage{amsfonts} 
\usepackage{amsmath} % for multi-lined equations
\usepackage{thmtools}
\usepackage{hyperref} % for linkable titles and referneces.
\usepackage{tabularx} % for side-by-side equations.

\usepackage{listings}
\usepackage{xcolor} % For code formatting (coloring).

\definecolor{codegreen}{rgb}{0,0.6,0}
\definecolor{codered}{rgb}{0.9,0,0}
\definecolor{codeorange}{rgb}{0.9,0.5,0.0}
\definecolor{codegray}{rgb}{0.5,0.5,0.5}
\definecolor{codepurple}{rgb}{0.58,0,0.82}
\definecolor{backcolour}{rgb}{0.95,0.95,0.92}

\lstdefinestyle{mystyle}{
    backgroundcolor=\color{backcolour},   
    commentstyle=\color{codered},
    keywordstyle=\color{codeorange},
    numberstyle=\tiny\color{codegray},
    stringstyle=\color{codegreen},
    basicstyle=\ttfamily\tiny,
    breakatwhitespace=true, % break lines only on whitespaces
    breaklines=true,                 
    captionpos=b,                    
    keepspaces=true,                 
    numbers=left,                    
    numbersep=5pt,                  
    showspaces=false,                
    showstringspaces=false,
    showtabs=false,
    tabsize=2,
    frame=single, %showing frame outside code (none/leftline/topline/bottomline/lines/single/shadowbox)
}
\usepackage{algorithm}

\usepackage{mathtools} % for matrices

\usepackage{subcaption}
\usepackage{enumitem} % for un-tabbing and items spacing in 'enumerate' and 'itemize' clauses.
\usepackage{pgfplots} % for colored text and for plotting graphs
\pgfplotsset{compat=1.6}

\theoremstyle{definition} % In this style the theorems/definitions/etc. are *not* italicized.
\newtheorem{theorem}{Theorem}[section]
\newtheorem{lemma}[theorem]{Lemma} 
\newtheorem{corollary}[theorem]{Corollary} 
\newtheorem{remark}[theorem]{Remark} 
\newtheorem{definition}[theorem]{Definition} 
\newtheorem{example}[theorem]{Example} 
 
\newtheorem{problem}{Problem} % separate problem numbering from anything else.
\newtheorem{property}{Property} % separate problem numbering from anything else.
\newtheorem{conjecture}{Conjecture}

\definecolor{myblue}{RGB}{80,80,160}
\definecolor{mygreen}{RGB}{80,160,80}

\newcommand{\largeTablecaption}[1]{#1 Following, \emph{italicized} refer to column titles. \emph{Frac}, \emph{Vs} and \emph{denoms.} are short for fractional, vertices and denominators. \emph{Topology $U_{(n,i)}$} correspond to Figure~\ref{figure_all_trees_up_to_n8}. All columns are with respect to $D$-space except for the last. Out of all \emph{Primary Directions} (Definition~\ref{definition_primary_direction_terminology}) only few determine \emph{False Facets}, which reveal \emph{Frac Vs}. These vertices all happen to be half-integer (\emph{D-space denom.} of $1$ and $2$). \emph{Frac Vs Classes} is the number 
of equivalence classes of non-STT vertices by automorphic symmetries. \emph{$XZD$-space denom.} lists denominators of coordinates in vertices of the LP polytope. We were able to enumerate all the vertices for topologies marked with ${*}$. For the rest, we sampled random vertices by solving the LP in random $XZD$-directions. Lists without ${*}$ may be incomplete: For example, every extension of $U_{(5,0)}$ should have third-integer vertices by Theorem~\ref{theorem_subgraph_with_frac_vertex_implies_frac_vertex_extension}, yet $U_{(7,3)}$ (allegedly) does not.}

\title{Search Trees on Trees via LP}
\author{
Yaniv Sadeh\footnote{Tel Aviv University, yanivsadeh@tau.ac.il. Partially supported by the Israel Science Foundation grant no.\ 1156/23 and the Blavatnik Family Foundation.}
\and
Haim Kaplan\footnote{Tel Aviv University, haimk@tau.ac.il. Partially supported by the Israel Science Foundation grant no.\ 1156/23 and the Blavatnik Family Foundation.}
\and
Uri Zwick\footnote{Tel Aviv University, zwick@tau.ac.il}
}

\begin{document}

\maketitle

\begin{abstract}
We consider the problem of computing optimal search trees on trees (STTs). STTs generalize binary search trees (BSTs) in which we search nodes in a path (linear order) to search trees that facilitate search over general tree topologies. Golinsky~\cite{GolinskyThesis} proposed a linear programming (LP) relaxation of the problem of computing an optimal static STT over a given tree topology. He used this LP formulation to compute an STT that is a $2$-approximation to an optimal STT, and conjectured that it is, in fact, an extended formulation of the convex-hull of all depths-vectors of STTs, and thus always gives an optimal solution.\footnote{To be exact, since only non-negative weights are concerned in the objective of the LP, this is not the convex hull itself but the Minkowski sum of the lower envelope of the convex hull with the positive orthant.} In this work we study this LP approach further. We show that the conjecture is false and that Golinsky's LP does not always give an optimal solution. 
To show this we use what we call the \emph{normals method}.
We use this method to enumerate over vertices of Golinsky's polytope for all tree topologies of no more than 8 nodes. We give a lower bound on the integrality gap of the LP and on the approximation ratio of Golinsky's rounding method. We further enumerate several research directions that can lead to the resolution of the question whether one can compute an optimal STT in polynomial time.
\end{abstract}

\section{Preliminaries, Definitions and Known Results}
\label{section_preliminaries} 

We consider the problem of search trees on trees (STTs) which we now describe briefly. For a broader context and additional references, see the works of \cite{BerendsohnThesis, SplayTreesonTrees_andPTAS, TangoSTT, GolinskyThesis}.

The problem of STTs is as follows. We are given a tree topology, and have to construct a search tree (STT) for it where we can search for nodes in response to requests. To answer a request, we make queries on the STT as follows. We start at the root, and query an oracle that tells us whether we found the searched node, or to which direction in the tree (a connected component), with respect to the query, we should proceed. Loosely speaking, the underlying tree and its search tree may have completely different structures, since it is usually beneficial to query a vertex which is not adjacent to $v$ following a query to $v$. But both are trees over the same set of nodes. The cost of a query is the number of oracle queries until the specified node is found.

We focus on the static setting, where the STT is fixed. In this case, we assume that we know the query distribution, $f$ (for frequencies), of the nodes in advance. We want to construct an optimal \emph{static} search tree for $f$. There is also the dynamic setting which we do not consider here, where the STT may be restructured by rotations (appropriately defined). Dynamic algorithms are described in~\cite{BerendsohnThesis, SplayTreesonTrees_andPTAS, TangoSTT}.

We study a specific approach to computing STTs, that is based on linear programming (LP). Our text is self-contained and introduces everything that is necessary for understanding our results. Before we proceed, we give important terminology and definitions, followed by a formal statement of the problem we study.

\begin{definition}[Basic Definitions and Notations]$ $
\label{definitions_basics}
%Basic definitions and notations:

\begin{enumerate}
    \item Node: will always refer to an element of a tree.
    \item Underlying Topology: The topology over which to search. The topology will be consistently denoted as $U$ (for \emph{Underlying}). We  only consider topologies that are trees.
    \begin{enumerate}
        \item We also abuse the notation of $U$ to denote the set of nodes in this tree. Moreover, $n$ will always denote the number of nodes in $U$: $n \equiv |U|$.
        \item For two nodes $i \ne j$, denote by $(i \leftrightsquigarrow j)$ the (unique) path between them in $U$, excluding $i$ and $j$. Use square brackets to include $i$ and $j$: $[i \leftrightsquigarrow j]$ (both), $[i \leftrightsquigarrow j)$ (with $i$), $(i \leftrightsquigarrow j]$ (with $j$).
        \item We refer to specific topologies with $2 \le n \le 8$ as $U_{(n,i)}$ where $i$ changes according to the topology. Figure~\ref{figure_all_trees_up_to_n8} (Section~\ref{section_counter_example_main_polytope_study_disproving_conjecture}) and Table~\ref{table_trees_edges} (Appendix~\ref{section_appendix_extras}) show the exact correspondence.
    \end{enumerate}
    
    \item STT: Search tree on tree. We use this term to avoid confusion with the underlying topology which is also a tree. We consistently use $T$ to denote an STT.
    \begin{enumerate}
        \item An STT $T$ is defined over a topology $U$ recursively as follows: Pick a node $r \in U$ and set it to be the root of $T$. Each subtree of $r$  in $T$ is an STT which is  constructed recursively over a connected component of $U \setminus \{ r \}$. See Figure~\ref{figure_topology_to_stt_example} for an example.
        
        \item Frequencies, denoted by $f$: the relative number of queries of each node, a vector of $n$ non-negative values that sum to $1$.
        \item $depth(v)$: The depth of a node $v$ in a given STT is the number of ancestors on its path to the root of the STT, including itself, $depth(root(T))=1$.
        \item $cost(T,f)$: the cost of a (static) STT $T$ with respect to frequencies $f$ is the  sum of the depths of the nodes weighted by their frequencies: $cost(T,f) = \sum_{v \in U} f_v \cdot depth(v)$.
    \end{enumerate}

    \item Linear programming (LP) related:
    \begin{enumerate}
        \item Weights, Direction, Objective (function): may be used interchangeably to describe the objective function of the linear program.
        \item Vertex: always refers to a vertex of the polytope of the linear program (not to be confused with tree nodes).
        %\item $LP(U,w)$: The solution of the LP that is constructed for a topology $U$, with objective (weights) $w$. We abuse notation and use it both for the vertex and its value.
        \item We say that a point is integer if all of its coordinates are integer, and that a polytope is integer if all of its vertices are integer.

        \item We reserve $P$ to denote points in space, and use subscripts such as $P_X$ to denote the projection of $P$ to a set of coordinates $X$. We use $\mathbb{P}$ to denote a collection of points or a polytope. When in need, we use additional capitalized letters to denote points, and their corresponding blackboard-bold font to denote a set of points, for example $S$ and $\mathbb{S}$.
    \end{enumerate}

    \item Cost versus Value: we only use \emph{cost} when we discuss the cost of an STT, and use \emph{value} when we consider the value of a feasible point of the LP with respect to an objective function.
\end{enumerate}
\end{definition}

\begin{figure}[t!]
\centering
    \begin{subfigure}{0.35\textwidth}
        \includegraphics[width=1.0\textwidth]{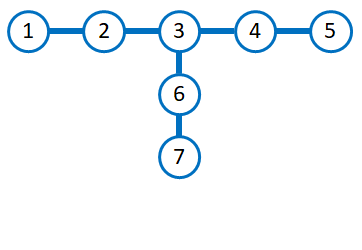}    
    \caption{Underlying tree $U$}
    \label{stt_example_U}
    \end{subfigure}
    ~
    \begin{subfigure}{0.24\textwidth}
        \includegraphics[width=1.0\textwidth]{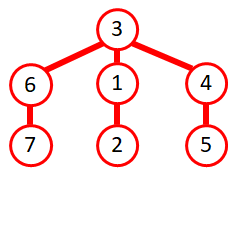}
        \caption{STT $T_1$}
        \label{stt_example_T1}
    \end{subfigure}
    ~
    \begin{subfigure}{0.22\textwidth}
        \includegraphics[width=1.0\textwidth]{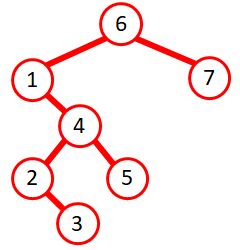}
        \caption{STT $T_2$}
        \label{stt_example_T2}
    \end{subfigure}
    
\caption{\footnotesize{We show two STTs $T_1$ and $T_2$ over the same topology $U$, as described in Definition~\ref{definitions_basics}. 
%There can be many different STTs that correspond to the same topology $U$, among which the simplest STTs are such that we simply root $U$ at some chosen node. In this illustration, the STTs are less trivial.
Note that while both $U$ and $T_i$ are trees on the same set of nodes, their edges may connect the nodes very differently.
%because an STT may ``jump'' between nodes of the topology.
}}
\label{figure_topology_to_stt_example}
\end{figure}

\begin{problem}[Static STT]
\label{problem_main_problem_stt}
Let $U$ be a tree topology, and let $f$ be a vector of frequencies. Determine an STT $T$ over $U$ such that $cost(T,f)$ is minimized.
\end{problem}

\subsection{The Linear Program}
\label{subsection_linear_program}

The work of \cite{GolinskyThesis} defines a linear program (LP) formulation to find a $2$-approximation for an optimal solution. The LP uses the following reasoning: Given a topology $U$ and the STT $T$ over it, for every pair of nodes, one must be an ancestor of the other, or there is an intermediate node between them that was queried first and acts as their LCA (least common ancestor) in the STT. We define an \emph{ancestry variable} $X_{ij}$ for each pair of nodes  $i,j \in U$ which should take the value $1$ if $i$ is an ancestor of $j$ in $T$, and another variable $Z_{kij}$ which should be $1$ if $k$ is the LCA of $i$ and $j$ in $T$. We only define $Z_{kij}$ for nodes $k \in U$ that are on the path between $i$ and $j$, i.e., $k \in (i \leftrightsquigarrow j)$. Ideally, one would hope to write $Z_{kij} = \min(X_{ki},X_{kj})$ but in an LP formulation we cannot do this, so we are content with writing the inequalities $Z_{kij} \le X_{ki}$ and $Z_{kij} \le X_{kj}$.\footnote{The LP we define soon does not capture additional STT properties, see further discussion in Section~\ref{section_refining_the_LP}.} We can also define the depth variable $D_i = \sum_{(i \ne) j \in U}{X_{ji}}$. The depth represents the cost of search, off by one, such that one aims to minimize $\sum_{i \in U}{f_i \cdot D_i}$. Ideally we would also require that all of $X$ and $Z$ are either $0$ or $1$, and this would define an integer LP (ILP). However, to be able to rely on LP machinery, we relax this definition and let the variables have non-negative real values.\footnote{Allowing values larger than $1$ is not a problem, because any quantity above $1$ can only hurt the minimization. Bounding each variable only introduces new constraints, so we prefer to ``let the LP do its job'' without it.}
%In addition, we also make another relaxation to be noted later.
To conclude this exposition, the final LP is as follows:

\begin{definition}[STT LP]
\label{definition_LP_program}
A complete description of the LP is as follows:
    \begin{enumerate}
        \item Variables: $X,Z,D$:
            \begin{enumerate}
                \item (Depths) $\forall i \in U: D_{i}$.
                \item (Ancestry) $\forall i,j \in U, i \ne j: X_{ij}$.
                \item (LCAs) $\forall i,j \in U, i \ne j: \forall k \in (i \leftrightsquigarrow j):  Z_{kij}$. Note that $Z_{kij}$ and $Z_{kji}$ are considered as the same variable. (Rather than defining both and equating $Z_{kij}=Z_{kji}$.)
            \end{enumerate}
        \item Constraints:
        \begin{enumerate}
            \item (Bounds\footnote{Technically we can omit the non-negativity constraints on $Z$, since $Z$ is not part of the objective and can only ``help'' the ancestry-constraints when it is non-negative. Alternatively, we can keep $Z \ge 0$ and have the non-negativity of $X$ implied by the loosely-LCA constraints. We cannot omit both sets, since then the LP becomes unbounded.}) $\forall i,j \in U, i \ne j: 0 \le X_{ij},X_{ji} \ \ \ \wedge \ \ \ \forall k \in (i \leftrightsquigarrow j): 0 \le Z_{kij}$.
            \item (Ancestry) $\forall i,j \in U, i \ne j: X_{ij} + X_{ji} + \sum_{k \in (i \leftrightsquigarrow j)}{Z_{kij}} \ge 1$.
            \item (Loosely LCA) $\forall Z_{kij}: (1) \; Z_{kij} \le X_{ki} \ \ \ ; \ \ \ (2) \; Z_{kij} \le X_{kj}$.
            \item (Depth) $\forall i \in U: D_i \ge \sum_{(i \ne) j \in U}{X_{ji}}$.
        \end{enumerate}
        \item Objective function: minimize $f \cdot D = \sum_{i \in U}{f_i \cdot D_i}$.
    \end{enumerate}
\end{definition}

Observe that we relaxed the ancestry and depth constraints to be inequalities, rather than equalities, even though STTs satisfy them with equalities.
%%% [Omitted for space considerations]
%%% On the one hand, equalities reduce the space of solutions, so it may be beneficial to keep them. On the other hand, since we already deal with a minimization problem it does not hurt to have inequalities, and in fact leaving this extra freedom will likely be safer when we later practically solve the LP numerically in some cases.

\begin{remark}
\label{remark_lp_variables_all_viewed_as_LCA}
$X$ and $Z$ seem to serve slightly different roles. A canonical way to look at them is by adding variables of the form $Z_{iij}$ to represent $i$ being the LCA of $i$ and $j$. In this case $X_{ij} = Z_{iij}$, and the ancestry inequality becomes $\forall i,j: \sum_{k \in [i \leftrightsquigarrow j]}{Z_{kij}} \ge 1$. If we wish, we can also define $Z_{iii}$ and have an ancestry inequality for $i=j$, which would imply $Z_{iii} \ge 1$ (every node is its own ancestor) and redefine $D_i \ge \sum_{j \in U}{Z_{jji}}$ so that the STT depth and the LP depth are the same. All of this being said, we stick to the separated presentation of the $X$ and $Z$ variables, particularly because $Z$ is a little more artificial (see Section~\ref{subsection_with_and_without_z_fourier_motzkin}).
\end{remark}

\subsection{Conjectures and Known Results}
In this subsection we summarize important facts regarding the LP, all were proven in \cite{GolinskyThesis}.

\begin{definition}[STT Induced Point]
\label{definition_xzd_of_stt}
Let $T$ be an STT. We denote by $(X^T,Z^T,D^T)$ the LP variables it induces: $X_{ij}=1$ if $i$ is ancestor of $j$, $Z_{kij}=1$ if $k = LCA(i,j)$ in $T$. The rest of the $X,Z$ variables are $0$, and $\forall i \in U: D_i = \sum_{j \in U\setminus \{ i \}}{X_{ji}}$.
\end{definition}

\begin{property}[Trees are Feasible]
\label{property_stt_feasible}
Let $T$ be an STT. Then $(X^T,Z^T,D^T)$ is a feasible solution. Moreover, the ancestry and depth inequalities are tight (equalities).
\end{property}

\begin{remark}[Depth Discrepancy]
\label{remark_depth_off_by_one}
Let $T$ be an STT. Observe that $\forall v \in U: depth(v) = D^T_v + 1$. This means that we can relate the cost of the STT to the value of the LP on the point it induces: $cost(T,f) = \sum_{v \in U}{f_v \cdot depth(v)} = \sum_v{f_v} + f \cdot D$. If $f$ is a vector of weights and not frequencies, first normalize it. Throughout the text, we usually deal with depths of the LP (the $D$ variables), and it is fine since we know how to relate them to the actual STT depths and its cost.
\end{remark}

\begin{property}[Integer Domination]
\label{property_integer_domination}
Let $(X',Z',D')$ be an \emph{integer} feasible solution. Then there exists an STT $T$, which can be computed efficiently, such that $\forall i \in U: D^T_i \le D'_i$. Alternatively phrased: the Integer-LP has an optimum that corresponds to an STT.
\end{property}

\begin{property}[$2$-Approximation]
\label{property_tree_2_approx}
Let $(X',Z',D')$ be \emph{any} feasible solution. Then there exists an STT $T$, which can be computed efficiently, such that $\forall i \in U: D^T_i \le 2 \cdot D'_i$.
\end{property}

Both Property~\ref{property_integer_domination} and Property~\ref{property_tree_2_approx} are proven constructively in \cite{GolinskyThesis}. The STT that achieves these guarantees is computed by the following rounding scheme. Note that this scheme not only rounds all the coordinates to be integer, it also takes integer solutions and ``rounds'' them to STTs.

\begin{definition}[Conversion to STT: Root Rounding]
\label{definition_rounding_scheme}
Let $U$ be a tree topology, and let $P = (X,Z,D)$ be the solution found by the LP for some weights. ${STT}_U(P)$ is the search tree on $U$ that is computed recursively as follows.
If $U$ has a single node, then we have a trivial singleton search tree. Otherwise, find a node $r$ such that: $\forall v (\ne r) \in U: \sum_{u \in U:r \in [u \leftrightsquigarrow v)}X_{uv} \ge \frac{1}{2}$. We clarify that $r \in [u \leftrightsquigarrow v)$ means that $r$ is on the path between $u$ and $v$, and may be $u$ (but not $v$). \cite{GolinskyThesis} proves that such a node exists. Set $r$ as the root of $STT_U(P)$, and construct its subtrees ${STT}_H(P)$  for each connected component $H \subseteq U \setminus \{r\}$ recursively. We refer to this rounding method as \emph{root rounding} or simply \emph{rooting}.
\end{definition}

Intuitively, the choice of the root guarantees that some of the (fractional) depth of each node $v$ is due to nodes that are $r$ or belong to a different connected component than the one of $v$, when $r$ is removed (so we charge $r$ for them). Because of this, setting $r$ to be the root increases the depth of any node $v \ne r$ by at most $\frac{1}{2}$ compared to
its fractional depth $D_v$ which we know by the choice of $r$ is at least $\frac{1}{2}$ for every $v \ne r$. For a formal proof that such $r$ always exists, see Lemma 3.4 in \cite{GolinskyThesis}. We emphasize that there may be multiple choices for the root at any given step, so the resulting STT and its approximation ratio may depend on the tie breaking rule.

Property~\ref{property_integer_domination} (Theorem 3.1 of \cite{GolinskyThesis}) implies that when considering integer solutions, there is an optimal solution that is indeed a search tree (in fact, it is a vertex of the LP). Property~\ref{property_tree_2_approx} (Theorem 3.2 of \cite{GolinskyThesis}) implies that the feasible point corresponding to an optimal STT is a $2$-approximation for the LP optimum. Golinsky \cite{GolinskyThesis} conjectured the following.

\begin{conjecture}[Conjecture 3.1 of \cite{GolinskyThesis}]
\label{conjecture_LP_finds_STT_vertex}
For any topology and any vector of non-negative weights, there is an optimal STT that induces an optimal solution to the LP in Definition~\ref{definition_LP_program}.
\end{conjecture}
   
Conjecture~\ref{conjecture_LP_finds_STT_vertex} implies that all the vertices of the lower envelope of the projection of the LP to $D$-space correspond to STTs, since by Theorem~\ref{theorem_integer_vertex_iff_stt} (in Section~\ref{section_positive_results}) every vertex has some objective for which it is uniquely optimal. In the language of polytopes and extended formulations~\cite{ExtendedFormulationSurvey}, Conjecture~\ref{conjecture_LP_finds_STT_vertex} states that the LP in Definition~\ref{definition_LP_program} is an extended formulation of an original polytope of interest. The original polytope is the convex hull of all STTs in $n$-dimensional $D$-space (depth per node), and the LP is (conjectured to be) an extended formulation with additional $X$ and $Z$ variables. To be exact, since only non-negative weights are concerned, our object is the Minkowski sum of the lower envelope of the convex hull with the positive orthant.

Since we care about Problem~\ref{problem_main_problem_stt} and the LP is just a tool, a weaker conjecture is as follows.

\begin{conjecture}
\label{conjecture_LP_rounding_finds_opt_STT}
For any topology and any vector of non-negative weights, rounding an optimal solution of the LP in Definition~\ref{definition_LP_program}, according to the root rounding scheme (Definition~\ref{definition_rounding_scheme}) yields an optimal STT.
\end{conjecture}

The work~\cite{GolinskyThesis} verified the conjectures for trees up to $n \le 4$ nodes, and lists several ideas and failed directions for a general proof. One hope to prove Conjecture~\ref{conjecture_LP_finds_STT_vertex} would be to show that all the vertices of the LP correspond to STTs, but this is not the case as we show in Theorem~\ref{theorem_non_tree_vertex}.

\textbf{We disprove both conjectures.} In simple words, we show that there exist tree topologies and frequencies such that the value of the best STT for them is strictly larger than the LP optimum, and that the rounding scheme might, and sometimes always, yield a sub-optimal STT. In particular, such optimum of the LP must be a non-integer solution by Property~\ref{property_integer_domination}.

\section{Disproving the Conjectures}
\label{section_counter_example_main_polytope_study_disproving_conjecture}
In this section we disprove Conjectures \ref{conjecture_LP_finds_STT_vertex} and \ref{conjecture_LP_rounding_finds_opt_STT}. We first disprove Conjecture~\ref{conjecture_LP_finds_STT_vertex} by showing a specific non-integer vertex for a specific topology and frequencies such that this vertex is strictly better than any STT (Section~\ref{subsection_specific_counter_example}). Then we generalize the discussion to elaborate further on our technique for finding this specific example and multiple others (Section~\ref{subsection_general_counter_exaple}), and study the integrality gap of the LP (Section~\ref{subsection_integrality_gap}). Finally, in Section~\ref{subsection_approximation_ratio} we disprove Conjecture~\ref{conjecture_LP_rounding_finds_opt_STT} and show that the proposed rounding scheme may result in a sub-optimal STT. There, we also  give a lower bound on the approximation-ratio for this specific rounding scheme.

\subsection{Specific Counter Example: Strictly Better Fractional Vertex}
\label{subsection_specific_counter_example}

In this section we provide an explicit example for a topology and frequencies such that the LP solution for it is strictly better than any solution induced by a search tree. The topology and frequencies were found using computer search (Appendix~\ref{appendix_section_code}), more details follow in Section~\ref{subsection_general_counter_exaple}. 
In this section we provide
a proof that the solution we found 
is a non-integer feasible solution to the LP, such that any other integer solution that is induced by a search tree is strictly worse.

%Regardless of anything that could go wrong with code such as possible bugs, numeric errors or incomplete exhaustive search, the example in this section is proven explicitly to be a non-integer feasible solution to the LP, such that any other integer solution that is induced by a search tree is strictly worse.

\begin{definition}[Terminology]
\label{definition_long_star}
For ease of reference, we define the following terminology. Let $U$ be a tree over $n=7$ nodes such that nodes $1$ to $5$ form a path in the natural order, with additional edges $(3,6)$ and $(6,7)$. This is also the tree in Figure~\ref{figure_topology_to_stt_example}, and $U_{(7,3)}$ in Figure~\ref{figure_all_trees_up_to_n8}. We refer to $U$ as a \emph{long star}, with three \emph{legs}. We refer to $3$ as its center, to $1,5,7$ as leaves, and to $2,4,6$ as mid-nodes.
\end{definition}

\begin{theorem}
\label{theorem_non_integer_optimum_long_star}
    Let $U$ be a long-star (Definition~\ref{definition_long_star}). Let $f = \frac{1}{23} [3,2,0,2,3,3,10]$ be the frequencies vector, i.e., such that $f_i$ is the frequency of querying node $i$. Then there is a feasible solution $P$ with depth coordinates $P_D = [2,2,4.5,2,2,1.5,0.5]$ that is strictly better for the LP with weights $f$ compared to any feasible point induced by an STT over $U$.
\end{theorem}

We note before the proof:
\begin{enumerate}
    \item Theorem~\ref{theorem_non_integer_optimum_long_star} holds even when we require a vector of strictly positive frequencies. Indeed, we can perturb $f_3$ by some $\epsilon > 0$ and still maintain a strictly-better feasible solution.
    
    \item $P$ is in fact a vertex, as determined by our LP-solver code. That being said, Theorem~\ref{theorem_non_integer_optimum_long_star} does not require $P$ to be a vertex, just that its value is strictly better than any STT.
\end{enumerate}

\begin{figure}[ht]
\includegraphics[width=1.0\textwidth]{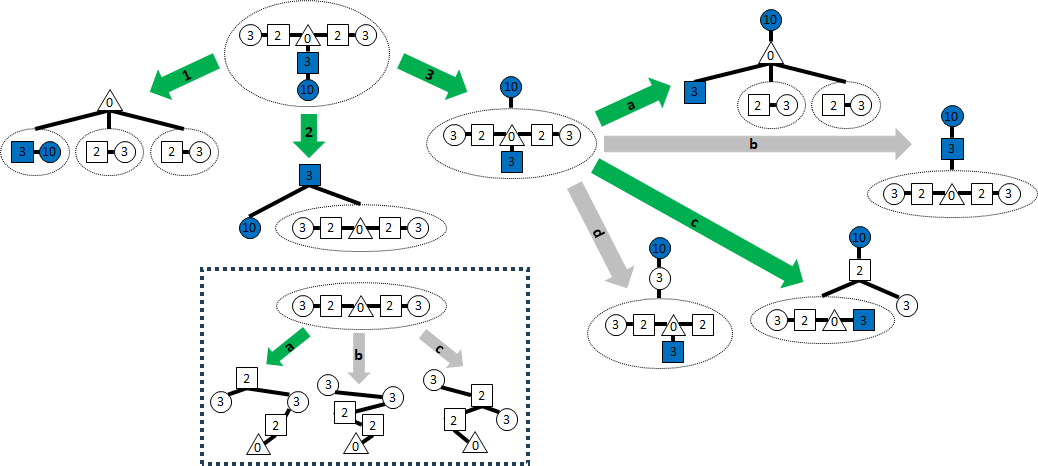}
\caption{\footnotesize{Visualization of the case analysis in proving Theorem~\ref{theorem_non_integer_optimum_long_star}. The blue nodes emphasize the heaviest leg, weights are written within the nodes, and the shapes emphasize the type of each node in the topology: center (triangle), mid-node (square), leaf (circle). A dashed ellipse indicates that the STT for this component is undetermined, yet. The numbers and letters within the arrows correspond to the cases of the proof. Green arrows lead to an STT that minimizes the LP among STTs, gray arrows lead to sub-optimal solutions.}}
\label{figure_case_analysis}
\end{figure}

\begin{proof}
For convenience, we scale the frequencies back to integer weights and use the vector $w = 23 \cdot f = [3,2,0,2,3,3,10]$.
%Its effect on the LP is just scaling, so nothing changes.
We divide the proof into two parts. First we describe all the coordinates of $P$ and verify its feasibility. Then we case-analyze the best possible value of a search tree, to show that it is strictly worse than the value of $P$.

\textbf{The solution $P$:} We present the values of the variables of $P$ in Figure~\ref{figure_explicit_counter_example}, where the $X_{ij}$ variables form a matrix $X$ (note that the variables $X_{ii}$ do not exist, so these entries are marked by '.'). The $D_i$'s are written below $X$, as they are the sums of the columns of $X$, and $Z_{kij}$ are to the side of $X$. The figure also verifies feasibility for each of the (in)equalities. The weighted depth of $P$ is: $P_D \cdot w = [2,2,4.5,2,2,1.5,0.5] \cdot [3,2,0,2,3,3,10] = 6+4+0+4+6+4.5+5=29.5$.

\begin{figure}[ht]
    {\footnotesize
    \begin{subfigure}{0.9\textwidth}
        \begin{equation}
        \label{equation_vertex_opt_values}
        \begin{bmatrix}
            X \\
            \hline
            D
            \end{bmatrix}
        =
            \begin{bmatrix}
            . & h & h & 0 & 0 & 0 & 0 \\
            h & . & 1 & h & h & h & 0 \\
            0 & 0 & . & 0 & 0 & 0 & 0 \\
            h & h & 1 & . & h & h & 0 \\
            0 & 0 & h & h & . & 0 & 0 \\
            h & h & 1 & h & h & . & h \\
            h & h & h & h & h & h & . \\
            \hline
            2 & 2 & 4.5 & 2 & 2 & 1.5 & 0.5 \\
            \end{bmatrix}
        ; Z:
        \begin{bmatrix*}[l]
            Z_{213}=Z_{214}=Z_{215}=Z_{415}=Z_{216}=Z_{617}=h \\
            Z_{314}=Z_{315}=Z_{316}=Z_{217}=Z_{317}=0 \\
    
            Z_{425}=Z_{627}=h \\
            Z_{324}=Z_{325}=Z_{326}=Z_{327}=0 \\
        
            Z_{435}=Z_{637}=Z_{647}=Z_{456}=Z_{657}=h \\
            Z_{346}=Z_{347}=Z_{356}=Z_{357}=Z_{457}=0 \\
        \end{bmatrix*}
        \end{equation}
        
    \end{subfigure}
    }
    
    {\footnotesize
    \begin{subfigure}{0.9\textwidth}
        \begin{equation}
        \label{equation_vertex_opt_inequalities}
        \begin{split} % See about split and alignment: https://www.overleaf.com/learn/latex/Aligning_equations_with_amsmath
        &
            \begin{bmatrix*}[l]
            X_{12}+X_{21} \\
            X_{13}+X_{31} + Z_{213} \\
            X_{14}+X_{41} + Z_{214} + Z_{314} \\
            X_{15}+X_{51} + + Z_{215} + Z_{315} + Z_{415} \\
            X_{16}+X_{61} + Z_{216} + Z_{316} \\
            X_{17}+X_{71} + Z_{217} + Z_{317} + Z_{617} \\
            X_{23}+X_{32} \\
            X_{24}+X_{42} + Z_{324} \\
            X_{25}+X_{52} + Z_{325}+ Z_{425} \\
            X_{26}+X_{62} + Z_{326} \\
            X_{27}+X_{72} + Z_{327} + Z_{627} \\
            \end{bmatrix*}
        =
            \begin{bmatrix*}[l]
            h+h \\
            0+h+h \\
            0+h+h+0\\
            0+0+h+0+h\\
            0+h+h+0\\
            0+h+0+0+h\\
            1+0\\
            h+h+0\\
            h+0+0+h\\
            h+h+0\\
            0+h+0+h\\
            \end{bmatrix*}
        ;
            \begin{bmatrix*}[l]
            X_{34}+X_{43} \\
            X_{35}+X_{53} + Z_{435} \\
            X_{36}+X_{63} \\
            X_{37}+X_{73} + Z_{637} \\
            X_{45}+X_{54} \\
            X_{46}+X_{64} + Z_{346} \\
            X_{47}+X_{74} + Z_{347}+ Z_{647} \\
            X_{56}+X_{65} + Z_{356}+ Z_{456} \\
            X_{57}+X_{75} + Z_{357}+ Z_{457}+ Z_{657} \\
            X_{67}+X_{76} \\
            \end{bmatrix*}
        =
            \begin{bmatrix*}[l]
            0+1\\
            0+h+h\\
            0+1 \\
            0+h+h\\
            h+h \\
            h+h+0 \\
            0+h+0+h\\
            0+h+0+h\\
            0+h+0+0+h\\
            h+h\\
            \end{bmatrix*}
        \\
        &
            \begin{bmatrix*}[l]
            Z_{213}=Z_{214}=Z_{215}=Z_{216} \\
            Z_{617}=Z_{627}=Z_{637}=Z_{647}=Z_{657} \\
            Z_{415}=Z_{425}=Z_{435}=Z_{456} \\
            \end{bmatrix*}
        =
        h
        \le
            \begin{bmatrix*}[l]
            X_{21},X_{23},X_{24},X_{25},X_{26} \\
            X_{61},X_{62},X_{63},X_{64},X_{65},X_{67} \\
            X_{41},X_{42},X_{43},X_{45},X_{46} \\
            \end{bmatrix*}
        =
            \begin{bmatrix*}[l]
            h,1,h,h,h \\
            h,h,1,h,h,h \\
            h,h,1,h,h \\
            \end{bmatrix*}
        \end{split}
        \end{equation}
    \end{subfigure}
    }
    
\caption{\footnotesize{The vertex $P$. (Eq.~\ref{equation_vertex_opt_values}) The values of $P$, where $h = \frac{1}{2}$ (short for half). The $Z$ variables are separated to three pairs of rows (of values $0$ and $h$). The first pair covers all $Z_{k1j}$, the next pair covers all $Z_{k2j}$ and the last pair covers $Z_{kij}$ for $i \ge 3$. The vector $D$ is a column-sum of the matrix $X$. (Eq.~\ref{equation_vertex_opt_inequalities}) Proof of feasibility, listing all the ancestry-inequalities. Inequalities of the form $Z_{kij} \le X_{ki},X_{kj}$ are verified in ``bulk'' (only non-trivial for $Z_{kij} > 0$).}}
\label{figure_explicit_counter_example}
\end{figure}

\textbf{Case analysis of all search trees:}
The number of STTs over $U$ is too large for hand verification. There are $662$ in total, but  $U$ is symmetric, and $w$ is also mostly symmetric. This enables us to reduce the analysis to a handful of cases. We emphasize that in the following discussion, the quality of an STT is considered with respect to $w = [3,2,0,2,3,3,10]$, and that the depths start from $0$ (root) (see Remark~\ref{remark_depth_off_by_one}).

Before the case analysis itself, we determine the best BST over the nodes $1$ to $5$, which is used later in two different cases. A best BST only on the subgraph of $U$ that contains nodes $1$ to $5$ has weighted depth of $10$, and depths: $[1,0,3,2,1]$. Indeed: in this case, we are dealing with a binary search tree (BST). Note that $w_3 = 0$, so we put it at the bottom, and we need to find the best BST over weights $w_{1,2,4,5}=[3,2,2,3]$. Thanks to symmetry, there are only three candidate options (other options are worse), see the dashed box in Figure~\ref{figure_case_analysis}: (a) The root has weight 2, with two children of weight $3$: then the weighted depth is $[3,2,2,3] \cdot [1,0,2,1] = 10$. (b) The root and its child have a weight of 3: then the weighted depth is $[3,2,2,3] \cdot [0,2,3,1] = 13$. (c) The root has weight 3, its child has weight $2$ and it is not its neighbor in the topology: then the weighted depth is $[3,2,2,3] \cdot [0,2,1,2] = 12$.

Finally, the case analysis. We show that the best weighted depth of an STT is $30$ (there are multiple STTs that attain it). This will conclude the proof, since we already determined that the weighted depth of $P$ is $29.5$. See also Figure~\ref{figure_case_analysis}.

\begin{enumerate}
    \item Consider an STT rooted at $3$: it divides the next search on $U$ into three separate $2$-node components (its legs). While there is a total of $8$ such STTs, there is a single dominating STT such that we root each subtree by the heavier node, which is a leaf in $U$. Overall, the weighted depth of this best tree is $w \cdot [1,2,0,2,1,2,1] = 3+4+0+4+3+6+10=30$.

    \item Consider an STT rooted at a mid-node: Then its adjacent leaf is a direct child of the root, and the other subtree is a BST on a path of $5$ nodes. Because all the legs are such that the mid-node is lighter than its corresponding leaf, and because the leg of $(6,7)$ is strictly heaviest, the best such STT must be rooted at $6$. The remaining subtree is therefore a BST over the path of nodes $1$ to $5$. We found earlier the best BST over this path. Therefore, adding $+1$ to the depth of each node in this BST due to it being a subtree, we get that the weighted depth of the best tree rooted at a mid-node is $w \cdot [2,1,4,3,2,0,1] = 6+2 + 0+6+6 + 0+10 = 30$.

    \item Consider an STT rooted at a leaf: Again by symmetry, it suffices to  consider only STTs rooted at $7$.
    \begin{enumerate}
        \item If the child of the root is $3$, then the remaining nodes divide again to independent legs, and we have a clear best STT among the four options. Its weighted depth is $w \cdot [2,3,1,3,2,2,0] = 6+6 +0+ 6+6 + 6+0 = 30$.

        \item If the child of the root is $6$, then its subtree is a BST over the path of nodes $1$ to $5$, and as we analyzed a best option has depths $[1,0,3,2,1]$ when rooted. In this case, the weighted depth is $w \cdot [3,2,5,4,3,1,0] = 9+4+0 + 8+9 + 3+0 = 33$.

        \item If the child of the root is $4$ (choosing $2$ is symmetric): then $5$ breaks-away to its own subtree, and the remainder is a BST over $\{1,2,3,6\}$. Since $3$ can be pushed down (zero weight), one can verify that the best BST over $\{1,2,6\}$ is rooted at $2$. Overall, we get in this case a weighted depth of $w \cdot [3,2,4,1,2,3,0] = 9+4+0 + 2+6 + 9 +0 = 30$.

        \item If the child of the root is $5$ (choosing $1$ is symmetric):
        \begin{enumerate}
            \item Rooting the sub-sub-tree at $v$ such that $v$ is either $1$, $4$ or $6$ is sub-optimal: it leaves the remaining nodes as a single connected component, and even if we assume that the depth of each of them is only $3$ (direct children of $v$), then the weighted depth is at least: $0 \cdot w_7 + 1 \cdot w_5 + 2 \cdot w_v + 3 \cdot (10 - w_v) = 33 - w_v \ge 30$.

            \item Rooting the sub-sub-tree at $3$: two options remain due to the leg of $(1,2)$, but the better weighted depth among them is $w \cdot [3,4,2,3,1,3,0] = 9+8+0 + 6+3 + 9+0 = 35$.

            \item Rooting the sub-sub-tree at $2$: while it leaves $\{3,4,6\}$ in the same connected component, $3$ has zero-weight and therefore can be pushed down, so the root in this sub-component would be $6$ followed by $4$ as its child. Overall, the best weighted depth in this case is $w \cdot [3,2,5,4,1,3,0] = 9+4+0 + 8+3 + 9 +0= 33$. \qedhere
            
        \end{enumerate}
    \end{enumerate}
\end{enumerate}
\end{proof}

\begin{remark}
\label{remark_integrality_gap_example_60over59}
In the proof of Theorem~\ref{theorem_non_integer_optimum_long_star} we found that $OPT(U,w) \le 29.5$ compared to the best feasible STT that gives $30$. This shows an integrality gap of $\frac{60}{59}$. (This gap is the same if we solve for $f$ instead of $w$, since the scaling cancels out.)
\end{remark}

\subsection{Counter Examples: A General Discussion and the Normals Method}
\label{subsection_general_counter_exaple}

In Section~\ref{subsection_specific_counter_example} we gave a very specific ``out of the blue'' example. Now we give a broader context on how we found it and how to find others using what we call the \emph{Normals Method}. Consider a fixed tree topology $U$. There are $N$ STTs over it, each inducing a different LP depths-vector, and all together span a convex polytope $\mathbb{P}$ in depth-space, that is, considering only $D$ out of $(X,Z,D)$. Because we have a minimization problem over the depths, and because all the weights are non-negative, we define and work with the polytope $\mathbb{P}' = \{A\ |\ \exists P \in \mathbb{P}: \forall i: A_i \ge P_i\}$.\footnote{$\mathbb{P}'$ is the Minkowski sum of $\mathbb{P}$ with the positive orthant.} Loosely speaking $\mathbb{P}'$ is the set of points dominated by $\mathbb{P}$ with respect to minimization.

Note that $\mathbb{P}$ and $\mathbb{P}'$ have the same vertices, because every vertex has a direction in which it uniquely minimizes the objective (See Theorem~\ref{theorem_integer_vertex_iff_stt}) so none is dominated by others. However, by trading the finite $\mathbb{P}$ with the infinite $\mathbb{P}'$ we removed every facet that was dominated by a vertex. See Figure~\ref{figure_illustrating_infinite_polytope_is_better}(a) versus \ref{figure_illustrating_infinite_polytope_is_better}(b).

\begin{figure}[ht]
    \begin{subfigure}{0.23\textwidth}
        \includegraphics[width=1.0\textwidth]{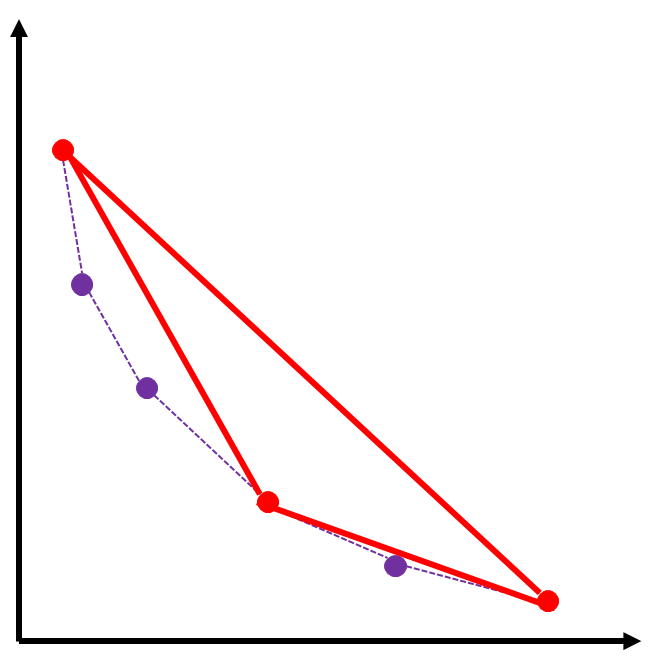}    
    \caption{Bounded polytope $\mathbb{P}$}
    \label{figure_illustrating_infinite_polytope_is_better_1polytope}
    \end{subfigure}
    ~
    \begin{subfigure}{0.23\textwidth}
        \includegraphics[width=1.0\textwidth]{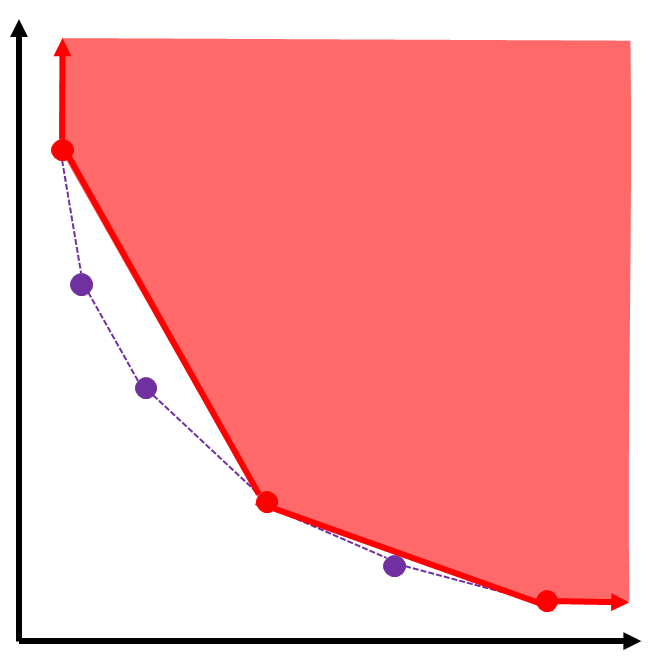}
        \caption{Infinite polytope $\mathbb{P}'$}
        \label{figure_illustrating_infinite_polytope_is_better_2envelope}
    \end{subfigure}
    ~
    \begin{subfigure}{0.23\textwidth}
        \includegraphics[width=1.0\textwidth]{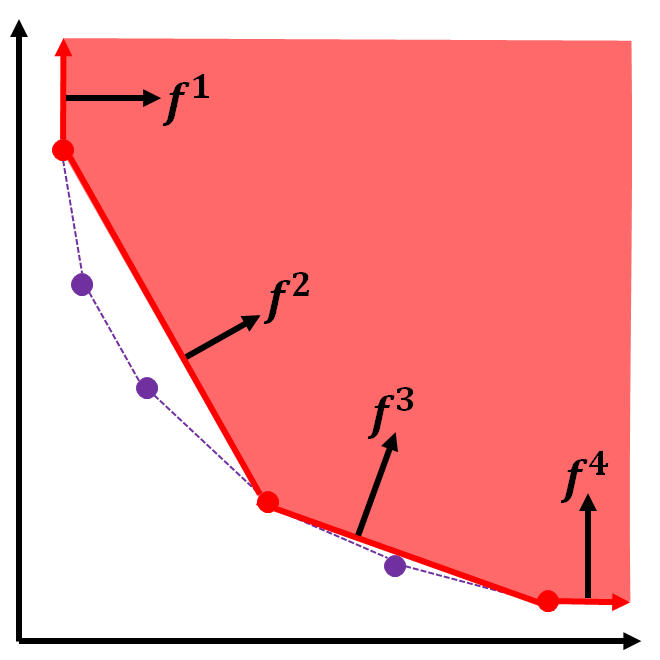}
        \caption{False facets}
        \label{figure_illustrating_infinite_polytope_is_better_3normal}
    \end{subfigure}
    ~
    \begin{subfigure}{0.23\textwidth}
        \includegraphics[width=1.0\textwidth]{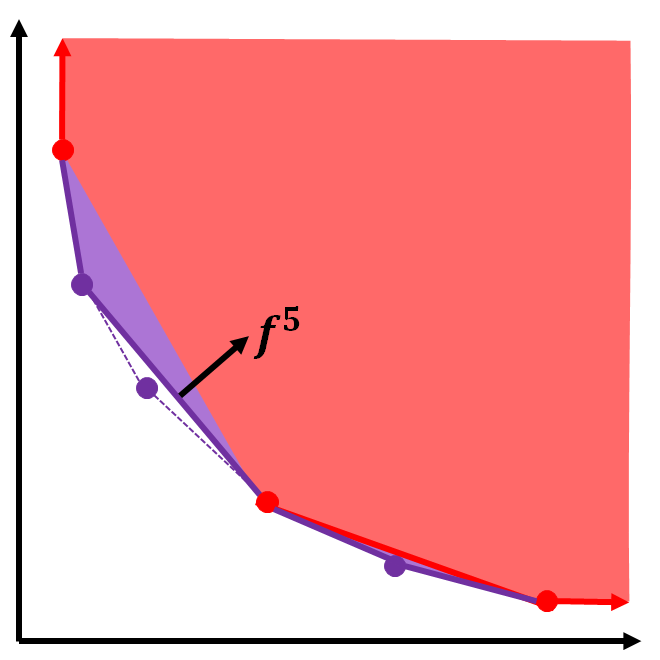}
        \caption{Iterative discovery}
        \label{figure_illustrating_infinite_polytope_is_better_4iterative}
    \end{subfigure}
\caption{\footnotesize{A 2-Dimensional illustration of the \emph{Normals Method} described in Section~\ref{subsection_general_counter_exaple}. (a) Shows 3 vertices which are (conceptually) due to STTs, that span a bounded red triangle, $\mathbb{P}$. Purple vertices denote additional (non-STT) vertices that we are unaware of, at this point. (b) shows the relaxed polytope $\mathbb{P}'$, of all points $A$ such that $\exists P \in \mathbb{P}: \forall i: P_i \le A_i$. $\mathbb{P}'$ has the same vertices as $P$, but only the facets of the lower-envelope of $\mathbb{P}$, as in this example. $\mathbb{P}'$ may gain axes-parallel facets. (c) We solve the LP in directions according to each normal vector $f^{i}$ to one of the facets of $\mathbb{P}'$. On a ``true'' facet the LP solution will not be better than expected, as for $f^{1}$ and $f^{4}$ in this illustration. However, if the facet is ``false'', the LP solution would be strictly better, revealing a new purple vertex, as for $f^{2}$ and $f^{3}$. Multiple normals may reveal the same vertex but only in $3$ dimensions and higher, see also Figure~\ref{figure_sage_3D_example}. By scanning all the facets of $\mathbb{P}'$ we determine if there are non-STT vertices, or not. (d) If we found new vertices, we can re-define the polytope and re-apply the method to search for additional vertices as exemplified by the normal $f^5$.}}
\label{figure_illustrating_infinite_polytope_is_better}
\end{figure}

\begin{figure}[ht]
    \centering
    \begin{subfigure}{0.4\textwidth}
        \includegraphics[width=1.0\textwidth]{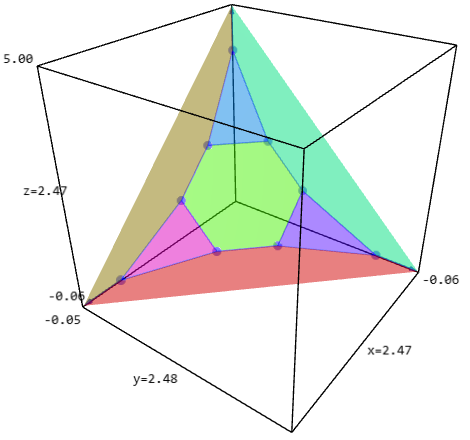}
    \caption{Without $[0.5,0.5,0.5]$}
    \end{subfigure}
    ~
    \begin{subfigure}{0.4\textwidth}
        \includegraphics[width=1.0\textwidth]{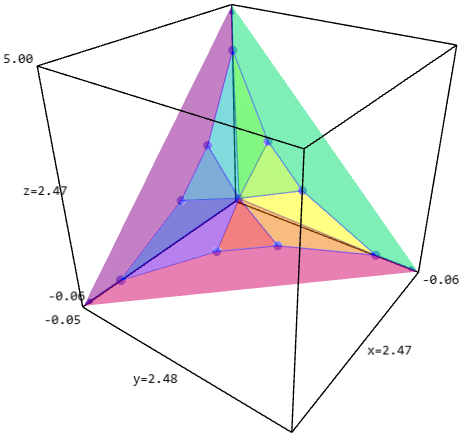}
        \caption{With $[0.5,0.5,0.5]$}
    \end{subfigure}
\caption{\footnotesize{A 3-dimensional illustration of how the \emph{Normals Method} described in Section~\ref{subsection_general_counter_exaple} may find a single vertex for multiple false facets. The figures were produced with \emph{Sage}. 
(a) $7$ Facets of the lower envelope of the convex hull of a set of $9$ integer vertices $A = \{(4,0,0),(0,4,0),(0,0,4),(0,1,2),(0,2,1),(1,0,2),(1,2,0),(2,0,1),(2,1,0)\}$.
(b) By adding to $A$ an extra vertex $(0.5,0.5,0.5)$, $4$ of the facets turn out to be false, while every vertex remains a vertex. In both figures, each facet has its own (arbitrary) color.} The vertices in this example are not related to STTs, since the smallest ``authentic'' example is $7$-dimensional.}
\label{figure_sage_3D_example}
\end{figure}

Given the $N$ known vertices induced by STTs, we can compute the facets of $\mathbb{P}'$, and solve the LP in the direction of the normal to each facet. If the facet is ``true'', then the solution would yield an optimum whose value is equal to the value of each of the STTs that span this facet. However, if the facet is ``false'', because it ``shaves off'' vertices that we ignored, then we will find a solution that is strictly better than any of the STTs that span the false facet, whose value is strictly better than the value of each of the STTs that span this facet, and this is why we work with $\mathbb{P}'$ rather than $\mathbb{P}$.  Observe that a false facet may hide many vertices, but every normal can reveal at most one new vertex when we solve the LP in its direction.\footnote{Assuming ``general position'', that is, no set of unknown vertices spans a face that is parallel to a false facet. Otherwise, a deterministic solver will still only find one vertex, but randomization may allow finding more.} We can theoretically find all the hidden vertices by iteratively refining the search, such that whenever a new vertex is found, we compute the new facets and solve the LP in the direction of the corresponding new normal vectors. We can repeat this until all the facets are determined to be true facets. Note that in the first iteration, every facet is due to STT vertices, hence there are multiple optimal STTs for the normal of a facet (exactly those that span the facet). This explains why in Section~\ref{subsection_specific_counter_example} our analysis determines multiple optimal STTs (revisit Figure~\ref{figure_case_analysis}), and more specifically, each point in the facet satisfies the equality $D \cdot w = 30$ where $D$ is the LP depths-vector and $w = [3,2,0,2,3,3,10]$ is the normal, which tells us in advance that the optimal LP value for the points induced by STTs is $30$. Note also that for $n \ge 3$, multiple normals may result in the same new vertex, as illustrated in Figure~\ref{figure_sage_3D_example}.

\begin{definition}[Primary Directions]
\label{definition_primary_direction_terminology}
We refer to each normal to a facet of the STT induced polytope as a \emph{primary direction}.
\end{definition}

Armed with the normals method, all that remains is to enumerate all the STTs, compute their depths-vectors as our $N$ vertices, convert from vertices to facets, enumerate the facets and for each solve the LP and check if the optimum is strictly better than was expected. This is not trivial, and is computationally demanding, therefore we were only able to run it for topologies of size $n \le 8$. We exhausted all possible topologies up to this size, and detail the high-level results in the caption of Figure~\ref{figure_all_trees_up_to_n8}. See Table~\ref{table_trees_summary_analysis_vertices_etc_FEW} and Table~\ref{table_trees_summary_analysis_vertices_etc_FULL} for the detailed results. The computation was done in \emph{Sage}, and in particular
we used sage to enumerate the facets of the polytope defined as the convex hull of the STTs.
%the conversion from vertices to facets was done blackbox-wise by defining the polytope by its collection of vertices and then enumerating over its facets. 
For more on the code, see Appendix~\ref{appendix_section_code}. Due to the computational complexity of converting vertices to facets, we were unable to finish an iterative scan for all the possible vertices. To clarify, we were able to fully scan all the primary directions for each topology, and determine new non-integer vertices in seven of them. However, we did not have enough resources to compute the new facets when taking into account the newly found vertices for the topologies with $8$ nodes.
%due to inability to allocate memory. 
This failed even when only $2$ new vertices were found. We were only able to run this second phase on the topology with $7$ nodes, and for it we determined that no additional new vertices exist. (The second phase had $6385$ normals, compared to $6364$ primary directions.)

In general, many of the vertices, and the primary directions, may be symmetric. Symmetries arise due to automorphisms of the topology. If $\pi: U \to U$ is an automorphism, then we can define an automorphism on the $(X,Z,D)$-space by mapping indices according to $\pi$. For example, it maps $D_i \to D_{\pi(i)}$ and $X_{ij} \to X_{\pi(i),\pi(j)}$. Since the ``names'' of the nodes are arbitrary, it is clear why a vertex $P$ translates by $\pi$ to another vertex $P'$, and we say that $P$ and $P'$ are \emph{symmetric} to each other. Note that $P$ and $P'$ have the same coordinate values, permuted by $\pi$. There are additional cases where primary directions can look symmetric, see Remark~\ref{remark_pseudo_symmetry}.

\begin{table}[!t]%[!ht]
    \scriptsize
    \begin{center}
    \begin{tabular}{|c|c|c|c|c|c|l|l|}

    \hline
    \begin{tabular}{@{}c@{}}Topology \\ $U_{(n,i)}$ \end{tabular} &
    STTs &
    \begin{tabular}{@{}c@{}}Primary \\ Directions \end{tabular} &
    \begin{tabular}{@{}c@{}}False \\ Facets\end{tabular} &
    \begin{tabular}{@{}c@{}}Frac \\ Vs\end{tabular} &
    \begin{tabular}{@{}c@{}}Frac Vs \\ Classes\end{tabular} &
    \begin{tabular}{@{}c@{}}$D$-space \\ denom. \end{tabular} &
    \begin{tabular}{@{}c@{}}$XZD$-space \\ denom. \end{tabular} \\
    \hline
    \hline
    (3,0) & \ \ \ \ 5 & \ \ \ \ 9 & \ \ 0 & . & . &\{1\} & \{1\} {*} \\
    \hline
    (4,0) & \ \ \ 14 & \ \ \ 32 & \ \ 0 & . & . &\{1\} & \{1\} {*} \\
    \hline
    (4,1) & \ \ \ 16 & \ \ \ 32 & \ \ 0 & . & . &\{1\} & \{1\} {*} \\
    \hline
    (5,0) & \ \ \ 42 & \ \ 145 & \ \ 0 & . & . &\{1\} & \{1, 2, 3\} {*} \\
    \hline
    (5,1) & \ \ \ 51 & \ \ 152 & \ \ 0 & . & . &\{1\} & \{1\} {*} \\
    \hline
    (5,2) & \ \ \ 65 & \ \ 161 & \ \ 0 & . & . &\{1\} & \{1\} {*} \\
    \hline
    (7,3) & \ \ 662 & \ 6364 & \ 39 & 9 & 2 & \{1, 2\} & \{1, 2\} \\
    \hline
    (8,4) & \ 2416 & 48291 & 362 & 65 & 38 & \{1, 2\} & \{1, 2, 3\} \\
    \hline
    (8,5) & \ 2952 & 56376 & 120 & 2 & 1 & \{1, 2\} & \{1, 2, 3\} \\
    \hline
    (8,6) & \ 2802 & 56724 & \ 10 & 2 & 1 & \{1, 2\} & \{1, 2, 3, 4\} \\
    \hline
    (8,11) & \ 3988 & 54201 & \ 78 & 18 & 4 & \{1, 2\} & \{1, 2\} \\
    \hline
    (8,12) & \ 3332 & 56404 & 528 & 60 & 24 & \{1, 2\} & \{1, 2, 3\} \\
    \hline
    (8,13) & \ 4076 & 65733 & 946 & 28 & 4 & \{1, 2\} & \{1, 2\} \\
    \hline
    \end{tabular}
    \end{center}
    \caption{
    \footnotesize{\largeTablecaption{Summary of all tree topologies up to $n \le 5$ nodes and those with non-STT $D$-space vertices up to $n \le 8$ (for more topologies, see Table~\ref{table_trees_summary_analysis_vertices_etc_FULL}).}}}
    \label{table_trees_summary_analysis_vertices_etc_FEW}
\end{table}

\begin{figure}
\includegraphics[width=1.0\textwidth]{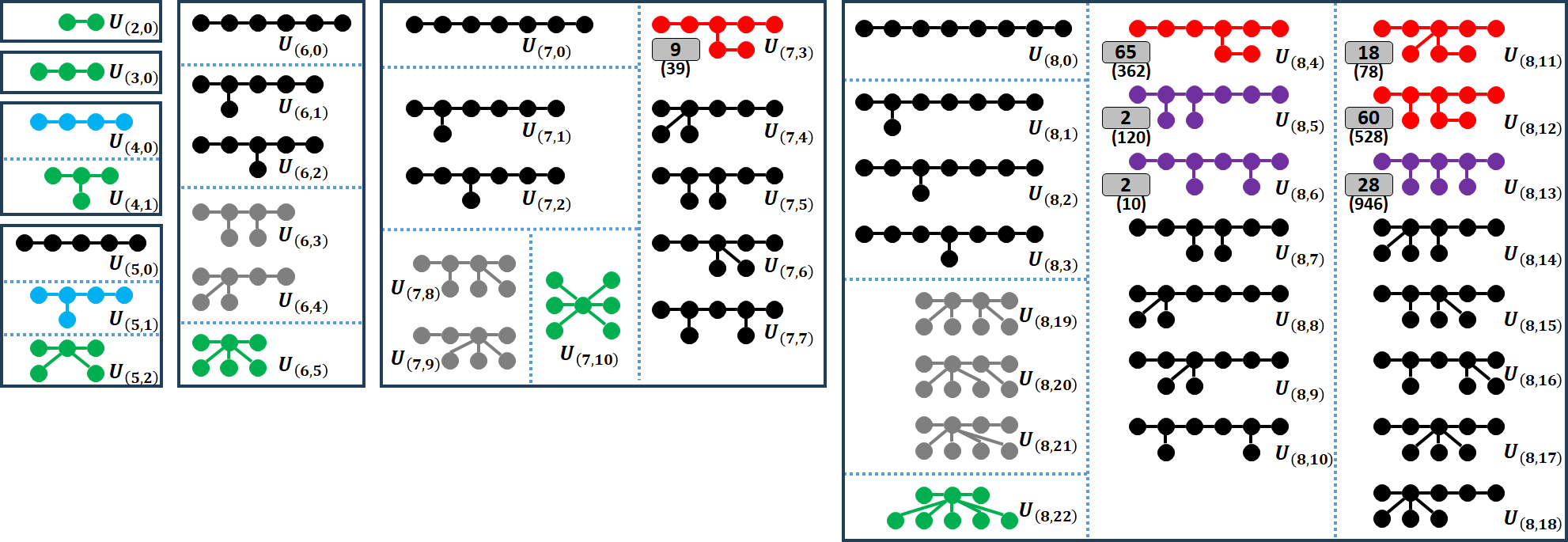}
\caption{\footnotesize{All tree topologies with $2 \le n \le 8$ nodes, 47 in total. Each tree is named $U_{(n,i)}$ where $n$ is its number of nodes and $i$ is a running index, over non-increasing diameter. For example, $U_{(n,0)}$ is always the path over $n$ nodes and the last among $U_{(n,i)}$ is the star with $n$ nodes. The number of trees for each $n$ is also known as OEIS sequence A000055 (see: https://oeis.org/A000055). Rectangles group by size $n$, and dashed lines divide by diameter within each group. Color scheme: Green trees are stars, analytically proven to have only integer vertices in $(X,Z,D)$-space. 
Blue trees were also verified to have only integer vertices in $(X,Z,D)$-space, with computer assistance. Black trees only have integer vertices in $D$-space, but have non-integer vertices in $(X,Z,D)$-space, this is further discussed in Section~\ref{subsection_with_and_without_z_fourier_motzkin}.
Gray trees may be either blue or black, it was infeasible to verify. Red and purple mark trees with non-integer vertices in $D$-space, i.e., counter-example vertices. The red trees are the unique such tree for $n=7$ and its extensions to $n=8$ (Definition~\ref{definition_extended_graph}), while purple trees are additional topologies that cannot be extended from the red tree with $7$ nodes. The number of non-integer vertices is written in a gray box near each tree, and in parenthesis is the number of false facets (or, primary directions) by which we found these vertices.}}
\label{figure_all_trees_up_to_n8}
\end{figure}

We follow with some additional information on the topology $U_{(7,3)}$ used in the proof of Theorem~\ref{theorem_non_integer_optimum_long_star}. Its LP has a total of $39$ false facets which correspond to $9$ non-integer vertices. By the automorphism symmetry of the topology it is natural that these numbers are divisible by $3$. There are $7$ normals, such that one of them has 3 symmetric copies and each of the other six has 6 copies. These normals are (the first was used in Theorem~\ref{theorem_non_integer_optimum_long_star}): $w = (3, 2, 0, 2, 3, 3, 10)$, $(14, 6, 0, 10, 32, 5, 7)$, $(16, 6, 0, 11, 34, 4, 8)$, $(39, 11, 0, 6, 21, 4, 8)$, $(18, 6, 0, 10, 36, 5, 7)$, $(18, 5, 0, 3, 6, 4, 5)$, $(9, 4, 0, 7, 22, 4, 5)$. Information about other trees is summarized in Table~\ref{table_trees_summary_analysis_vertices_etc_FEW}.

We conclude our findings with an interesting remark.

\begin{definition}[Partially Integer Vertex]
\label{definition_partially_integer_vertex}
We say that an $(X,Z,D)$-vertex is \emph{partially integer} if its $D$ coordinates are all integer, but at least one $X$ coordinate is non-integer.
\end{definition}

\begin{remark}[New vertices are half-integer]
\label{remark_new_vertices_are_non_integer_half_integer_only}
Partially integer vertices exist, an example is given in Figure~\ref{figure_Example_partially_interger_vertex}. Therefore, the normals method may technically find partially integer vertices which do not correspond to STTs. Indeed, while Property~\ref{property_integer_domination} says that STT-vertices dominate other integer points, it is with respect to fully integral points in $(X,Z,D)$-space.

However, interestingly every new $D$-space vertex that we were able to find has non-integer depths. Partially integer vertices that we found are always projected to non-vertices in $D$-space. Moreover, all depths are half-integer, that is, have coordinates that are $0$ or $\frac{1}{2}$, modulo $1$ (some coordinates may be integer, but not all of them). Is all of this an artifact of small sizes ($n \le 8$), like the fact that all the polytopes are integer for $n \le 4$, or is it an inherent property?
\end{remark}

\begin{example}
\label{figure_Example_partially_interger_vertex}
The smallest example for a partially integer vertex (Definition~\ref{definition_partially_integer_vertex}) is for $n=6$, see $P_1$ in Equation~(\ref{equation_partially_integer_vertex}). We find $P_1$ when solving the LP for the path topology, in some general $(X,Z,D)$-direction.\footnote{
The order of LP variables is defined in the code, in the function named ``\_construct\_dictionaries\_PRIMAL'' of the \emph{SearchTreeUtilities} class. The direction according to their order is [0, 5, 5, 1, 2, 3, 1, 4, 1, 1, 3, 3, 1, 4, 2, 2, 3, 5, 3, 1, 1, 3, 5, 4, 2, 2, 6, 0, 2, 1, 5, 3, 3, 1, 5, 4, 4, 5, 3, 2, 1, 6, 3, 5, 2, 3, 4, 5, 3, 2, 2, 1, 0, 3, 5, 2].} The depths-vector of $P_1$ is dominated by the STT (BST) $[2,1,2,0,1,2]$. In other cases (not shown) the depths are dominated by a convex combination of more than one STT. For $n \le 5$, only the path topology has non-integer vertices, and each has a non-integer $D$ coordinate (and is projected to a non-vertex in $D$-space).

{\scriptsize
\begin{equation}
\label{equation_partially_integer_vertex}
P_1 = \frac{1}{2} \cdot 
\begin{bmatrix}
. & 2 & 0 & 0 & 0 & 0 \\
0 & . & 1 & 0 & 0 & 2 \\
2 & 1 & . & 1 & 1 & 0 \\
1 & 1 & 1 & . & 1 & 2 \\
1 & 0 & 0 & 1 & . & 2 \\
0 & 0 & 2 & 0 & 0 & . \\
\hline
4 & 4 & 4 & 2 & 2 & 6 \\
\end{bmatrix}
% %% Another example, in the larger topology $U_{8,6}$ that extends PATH(6), and also has non-integer vertices. No need to show it, saving some space and cluttering.
% ~
% Q_2 = \frac{1}{2} \cdot 
% \begin{bmatrix}
%  . & 1 & 0 & 0 & 0 & 1 & 1 & 2 \\
% 1 & . & 1 & 1 & 1 & 0 & 1 & 0 \\
% 1 & 1 & . & 1 & 1 & 1 & 1 & 1 \\
% 0 & 1 & 1 & . & 1 & 1 & 1 & 1 \\
% 0 & 0 & 0 & 1 & . & 1 & 0 & 1 \\
% 0 & 0 & 0 & 0 & 1 & . & 0 & 0 \\
% 0 & 1 & 0 & 0 & 0 & 0 & . & 1 \\
% 0 & 0 & 0 & 1 & 0 & 0 & 0 & . \\
% \hline
% 2 & 4 & 2 & 4 & 4 & 4 & 4 & 6 \\
% \end{bmatrix}    
\end{equation}
}
\end{example}

\subsection{Integrality Gap - Lower Bounds}
\label{subsection_integrality_gap}
A lower bound of $\frac{60}{59}$ on the integrality gap of the LP follows from Theorem~\ref{theorem_non_integer_optimum_long_star}. For a tighter lower bound on the integrality gap, we can study all the non-integer vertices in all seven small topologies that have such vertices.

\begin{lemma}
\label{lemma_additive_gap}
Given a fixed topology $U$, let $\mathbb{D} = \{ D^1,\ldots,D^N \} $ be all the $D$-space vertices induced by STTs, and let $\mathbb{S} = \{S^1,\ldots,S^M \}$ be additional vertices of the $D$-space projection of the LP polytope, where each $S^i$ is found by solving the LP in the primary directions $h^{i,1},\ldots,h^{i,{k_i}}$ for some $k_i \ge 1$. Denote $\Delta \equiv \max \{ \min_{A \in conv^+(\mathbb{D})}(A \cdot f) - \min_{B \in conv^+(\mathbb{D} \cup \mathbb{S})} (B \cdot f) \mid \forall i: f_i \ge 0 \wedge \sum_{i=1}^{n}{f_i} = 1 \}$ (the largest additive gap), where $conv^+$ is a convex combination over a set of points, and positive rays.
Then if $\Delta > 0$, it is achieved for some $i,j,k$ such that $B=S^i$, $f = \hat{h}^{i,j}$ and $A = D^k$, where hat represents normalization to a sum of one, and $D^k$ is on the false facet of $conv^+(\mathbb{D})$ whose normal is $h^{i,j}$. (Recall Section~\ref{subsection_general_counter_exaple}.)
\end{lemma}
\begin{proof}
Let $f$ be a direction (vector) that maximizes $\Delta$. The scalar product with $f$ projects the polytope $conv^+(\mathbb{D})$ to one dimension, therefore the minimum has a pre-image that is a vertex, this is $A = D^k$. Similarly, by arguing for $conv^+(\mathbb{D} \cup \mathbb{S})$ we get that $B$ can be chosen as either some $D^i$ or $S^i$. However, $D^k \cdot f \le D^i \cdot f$ for every $1 \le i \le N$, so it must be that $B = S^i$ (or else $\Delta = 0$).

Now fix $A$ and $B$, denote $C = A-B$, and let us show that we can replace $f$ by some $\hat{h}^{i,j}$, without decreasing $\Delta$ which will conclude the proof. Since $f$ is a (normal to a) separating plane that intersects $A$, in a problem of covering form,\footnote{An LP problem of the form minimize $c \cdot x$ where $Ax \ge b$, $x \ge 0$, $c \ge 0$ for vectors $x,b,c$ and matrix $A$.} by strong-duality we get that $f = \sum_{j}{\alpha_j \cdot \hat{h}^j}$ for non-negative coefficients $\alpha_j$, where $h^j$ are all the normals to facets of $conv^+(\mathbb{D})$ that contain $A$. Moreover, $\sum_j {\alpha_j} = 1$ because all of $f$ and $\hat{h}^j$ have coordinate sum of $1$. We get that:
$C \cdot f = \sum_{j}{\alpha_j \cdot (C \cdot \hat{h}^j)} \le C \cdot \hat{h}^{j^*}$ by choosing $j^*$ that maximizes the scalar product. However, since $f$ is a maximizer, we get equality, and conclude that $\hat{h}^{j^*}$ is a primary direction that maximizes $\Delta$. Finally, it must be that $\hat{h}^{j^*}$ is a normal to a false facet that contains $A$. Otherwise, it does not imply a separating plane, and therefore $A \cdot \hat{h}^{j^*} - B \cdot \hat{h}^{j^*} \le 0$, which is a contradiction.
\end{proof}

By Lemma~\ref{lemma_additive_gap} it suffices to consider only vertices and primary directions to maximize the additive gap, and in particular only check for gaps of the form: $(D^k - S^i) \cdot h^{i,j}$. By choosing for each topology the vertices $S^i$ and $D^k$, and the primary direction $h^{i,j}$ that maximize $\Delta$, we get Table~\ref{table_integrality_gaps}. Anecdotally, the largest gap for $U_{(7,3)}$ is due to the primary direction which we studied in Theorem~\ref{theorem_non_integer_optimum_long_star}. Based on the small topologies we studied, the integrality gap is at least $\frac{95}{93} \approx 1.0215$, which is much closer to $1$ (almost optimal) than to the proven upper bound of $2$. Note that maximizing the difference does not necessarily maximize the integrality gap (ratio).
% So there is room for improving the lower-bound, even on the small topologies with $n \le 8$.

%%% The following is an LP to optimize the difference, but indeed, the optima are always primary directions.
% \begin{definition}[LP for Gap Optimization]
% \label{definition_LP_for_integrality_gap}
% We assume knowledge of all the STT depths-vectors, denotes $D^1,\ldots,D^N$, and another (non-STT) vertex $P$. We determine a frequencies vector $f$ that maximizes the difference: $\min_{i}\{D^i \cdot f\} - P \cdot f$.
% \begin{enumerate}
%     \item Variables: $x$ (scalar) and $f$ (vector).
%     \item Constraints: (a) $\forall 1 \le i \le n: f_i \ge 0$; (b)  $\sum_{i=1}^{n}{f_i} = 1$; (c) $\forall 1 \le i \le N: x \le D^i \cdot f$.
%     \item Objective: maximize $x - P \cdot f$.
% \end{enumerate}
% \end{definition}

Another interesting note is that the topology $U_{(8,11)}$ extends $U_{(7,3)}$ (Definition~\ref{definition_extended_graph}), and it happens to be that the direction that maximizes the difference for $U_{(8,11)}$ is an extension of the direction that maximizes the difference for $U_{(7,3)}$.

\begin{table}[!ht]%[!ht]
    \scriptsize
    \begin{center}
    \begin{tabular}{|c|c|c|c|r|}

    \hline
    Topology $U_{(n,i)}$ & Direction & LP Value & STT Best Value & Integrality Gap \\
    \hline
    (7,3) & (3,2,0,2,3,3,10) & \ 29.5 & \ 30 & $60/59 \approx 1.0169$ \\
    \hline
    (8,4) & (9,5,0,6,11,17,5,9) & \ 93 \ \ & \ 95 & $\mathbf{95/93 \approx 1.0215}$ \\
    \hline
    (8,5) & (16,2,3,6,7,13,34,5) & 121 \ \ & 122 & $122/121 \approx 1.0083$ \\
    \hline
    (8,6) & (55,1,3,4,14,29,34,8) & 200.5 & 201 & $402/401 \approx 1.0025$ \\
    \hline
    (8,11) & (3,2,0,2,3,0,3,10) & \ 29.5 & \ 30 & $60/59 \approx 1.0169$ \\
    \hline
    (8,12) & (11,1,0,1,2,6,1,2) & \ 28.5 & \ 29 & $58/57 \approx 1.0175$ \\
    \hline
    (8,13) & (7,1,1,1,7,7,2,7) & \ 49.5 & \ 50 & $100/99 \approx 1.0101$ \\
    \hline
    \end{tabular}
    \end{center}
    
    \caption{\footnotesize{For each topology with non-STT vertices, we show a primary direction (Definition~\ref{definition_primary_direction_terminology}) that maximizes the integrality gap among all primary directions. The largest gap overall is marked in bold.}}
    \label{table_integrality_gaps}
\end{table}

\subsection{Approximation Ratio of the Root Rounding - Lower Bounds}
\label{subsection_approximation_ratio}

So far we just disproved the first (and stronger) conjecture, and studied the integrality gap of the LP. In this section we return to the original problem of STTs, and consider the rounding in Definition~\ref{definition_rounding_scheme}. We disprove Conjecture~\ref{conjecture_LP_rounding_finds_opt_STT} that claims it rounds to an optimal STT. The rounding scheme may have degrees of freedom, and we study it both when assuming the worst-case scenario (pick a worst choice) and in the best-case scenario (pick a best choice).

Interestingly, for all the seven small topologies ($n \le 8$) for which non-STT vertices were found, when solving the LP in a primary direction, there is a rounding that yields an optimal STT. At first glance this suggests that there is hope to refine the rounding scheme to remove misleading degrees of freedom (of choosing which node to root at each iteration). However, we show later in this section that  in non-primary directions we are no longer guaranteed an optimal rounding. This ``strongly'' disproves Conjecture~\ref{conjecture_LP_rounding_finds_opt_STT} by showing that the rounding scheme could be suboptimal regardless of the degrees of freedom in the rounding procedure.

Table~\ref{table_approximation_ratios} summarizes our findings for the seven topologies with non-STT vertices. We show that every topology has directions for which even the best-case rounding is sub-optimal. We also show directions with large worst-case ratios, the largest that we found, though not by an exhaustive search or analysis, is $\frac{263}{190} \approx 1.384$. This is a lower-bound on the approximation ratio of the root rounding approach.

\begin{table}[!ht]%[!ht]
    \scriptsize
    \begin{center}
    \begin{tabular}{|c|c|c|c|c|c|c|c|c|c|}

    \hline
    Topology $U_{(n,i)}$ &  \multicolumn{2} {|c|} {Direction} & OPT STT & BC STT & WC STT & BC Ratio & WC Ratio \\
    \hline
    (7, 3) & (11, 7, 0, 10, 34, 7, 11) & {*} & 184/80 & 186/80 & 220/80 & 1.0109 & 1.1957 \\
    \hline
    (8, 4) & (17, 9, 0, 10, 19, 29, 9, 17) & {*} &  277/110 & 281/110 & 339/110 & \textbf{1.0144} & 1.2238 \\
    \hline
    (8, 5) & (13, 1, 2, 4, 5, 10, 25, 4) & {*} &  154/64 & 155/64 & 178/64 & 1.0065 & 1.1558 \\
    \hline
    (8, 6) & (86, 1, 5, 6, 22, 46, 55, 13) & {*} &  552/234 & 553/234 & 658/234 & 1.0018 & 1.192 \\
    \hline
    (8, 11) & (11, 7, 0, 7, 11, 0, 10, 34) & {*} &  184/80 & 186/80 & 220/80 & 1.0109 & 1.1957 \\
    \hline
    (8, 12) & (32, 2, 0, 3, 7, 18, 3, 7) & {*} &  160/72 & 162/72 & 188/72 & 1.0125 & 1.175 \\
    \hline
    (8, 13) & 
    \begin{tabular}{@{}l@{}}(48, 6, 5, 6, 48, 48, 15, 48) \\ $+\epsilon \cdot (-1,1,1,1,-1,-1,-1,-1) $ \end{tabular}
    & {**} &  $\frac{558-\epsilon}{224 - 2\epsilon}$ & $\frac{563-6\epsilon}{224 - 2\epsilon}$ & $\frac{647-7\epsilon}{224 - 2\epsilon}$ & 1.0090 & 1.1595 \\
    \hline
    (8, 4) & (6, 3, 0, 5, 0, 18, 4, 5) & ${p}$ &  94/41 & 94/41 & 129/41 & 1 & 1.3723 \\
    \hline
    (8, 4) & (6.5, 3, 0, 5, 0, 18, 4, 5) & ${p'}$ &  190/83 & 191/83 & 263/83 & 1.0053 & \textbf{1.3842} \\
    \hline
    (8, 13) & (10, 2, 1, 2, 35, 35, 2, 10) & ${p}$  & 216/97 & 216/97 & 283/97 & 1 & 1.3102 \\
    \hline
    \hline
    \end{tabular}
    \end{center}
    \caption{\footnotesize{For each topology whose LP finds non-STT vertices, we show directions and approximation ratios of the best-case (BC) and worst-case (WC) root rounding. The first $7$ rows give one entry per topology with a direction in which even the best-case is sub-optimal. The next rows show additional directions for which the worst-case is among the largest that we found. Directions marked ${*}$/${**}$ were found by the LP in Definition~\ref{definition_LP_for_approximation_gap}, ${p}$ marks primary directions that give large ratios, and ${p'}$ marks a small variation to increase the largest gap we found among primary directions. Remark~\ref{remark_int_gap_zero_must_perturb} explains why in the case of ${**}$ we perturb the nice integer direction by $\epsilon$.}}
    \label{table_approximation_ratios}
\end{table}

In the remainder of this section we detail the following: Theorem~\ref{theorem_rounding_the_vertex} demonstrates an explicit worst-case rounding of the optimum shown in Figure~\ref{figure_explicit_counter_example} (for Theorem~\ref{theorem_non_integer_optimum_long_star}), that gives an approximation ratio of $\frac{62}{53} \approx 1.170$. Then we explain how to find directions in which the root rounding is sub-optimal and where the directions in Table~\ref{table_approximation_ratios} come from.

\begin{theorem}
\label{theorem_rounding_the_vertex}
The approximation ratio of root rounding (Definition~\ref{definition_rounding_scheme}) is at least $\frac{62}{53}$.\\ (As shown in Table~\ref{table_approximation_ratios}, other (numerical) examples show that it can be even worse.)
\end{theorem}

\begin{proof}
Consider the case of Theorem~\ref{theorem_non_integer_optimum_long_star}, of weights $w = (3,2,0,2,3,3,10)$ and the non-integer solution described in Figure~\ref{figure_explicit_counter_example}. We show a set of choices that round this solution to an STT with depths-vector $D = (4, 3, 2, 3, 4, 1, 0)$. Then its LP value is $w \cdot D = 12+6+0+6+12+3 = 39$, compared to the best STT whose value is $30$. By Remark~\ref{remark_depth_off_by_one}, the approximation ratio is $\frac{39 + \sum{w_i}}{30 + \sum{w_i}} = \frac{62}{53} \approx 1.170$.

We now round step by step, for each connected component $C$ throughout the process.
\begin{enumerate}
    \item Initially, $C = \{1,\ldots,7\}$ (whole tree): since $X_{7i} = \frac{1}{2}$ for each $i=1,\ldots,6$, we can choose the root $r=7$. $C \setminus \{r\}$ remains a single connected component.

    \item $C = \{1,\ldots,6\}$: since $X_{6i} \ge \frac{1}{2}$ for each $i=1,\ldots,5$, we can choose the root $r=6$. $C \setminus \{r\}$ remains a single connected component, now it is a path.

    \item $C = \{1,\ldots,5\}$: while $X_{3i} = 0$ for $i=1,2,4,5$, choosing $r=3$ as the root satisfies the condition because $X_{24}=X_{25}=X_{41}=X_{42}=\frac{1}{2}$. $C \setminus \{r\}$ now has two connected components.

    \item $C = \{1,2\}$: since $X_{21} = \frac{1}{2}$ we can root at $2$. This leaves a singleton subtree $C \setminus \{2\} = \{ 1 \}$.

    \item $C = \{4,5\}$: since $X_{45} = \frac{1}{2}$ we can root at $4$. This leaves a singleton subtree $C \setminus \{4\} = \{ 5 \}$.
\end{enumerate}
Combining all of the above yields an STT whose  depths-vector $(4, 3, 2, 3, 4, 1, 0)$.
\end{proof}

We could refine root rounding by solving the LP recursively for each connected component after picking the root, instead of rounding always based on the same one-time solution. In the example of Theorem~\ref{theorem_rounding_the_vertex} it does help, because $7$ is a root of some optimal STT with respect to $w$, and we know that when the topology has size $n \le 6$ the LP approach is optimal. However, this refined rounding does not generally circumvent the problem, since we show next that we might pick the root sub-optimally. When this happens, the resulting STT is suboptimal even if we can improve our choices and the overall approximation ratio by constructing the subtrees better.

\medskip

In the remainder of this section, we discuss how to find directions 
for which root rounding has a
relatively large approximation ratio. First let us discuss the best-case rounding. As noted, to ensure a gap, even under the refined (iterated) root rounding, it is sufficient to ensure an empty intersection between the set of candidate roots, which we denote by $R$, and the set of roots of optimal STTs, which we denote by $O$.\footnote{Note that if we could guarantee a rounding that shares a root with an optimal solution, the refined (iterated) root rounding would give an optimal STT.} It  requires  to solve an LP $O(n)$ times instead of once, but the running time will remain polynomial.

As a first step, we can look for a primary direction for which $O \setminus R \ne \emptyset$. Then we can variate slightly the primary direction to ensure that all the optimal STTs have a root from $O \setminus R$. This naive approach does not always work, and by checking the sets $O$ and $R$ for each primary direction, we find multiple such directions only for the topologies $U_{(8,4)}$ and $U_{(8,13)}$.
%[extra details, omitted for brevity] Interestingly, in the case of $U_{(8,4)}$ whenever $O$ and $R$ differ we have that $O = \{3,4,6\}$, $R = \{4,5,6\}$. This is just a coincidence, for $U_{(8,13)}$ there are multiple different cases of partial intersections.
Moreover, primary directions are only special because they let us find a non-STT vertex, but they do not necessarily maximize the approximation ratio. To be more comprehensive, we can define an LP to find the unknown direction $f$ that maximize the distance between a non-STT vertex $P$ and all the STTs that root rounding can reach from it. Formally:

\begin{definition}[LP for Best-Case Approximation Ratio]
\label{definition_LP_for_approximation_gap}
Assume knowledge of all the depths-vectors of STTs, denoted by $D^1,\ldots,D^N$, and another (non-STT) vertex $P$. Also, let $S^1,\ldots,S^M$ be the set of all STTs we can get by root roundings of $P$. Ideally we aim to determine a frequencies vector $f$ that maximizes the separation: $\min_{i}\{S^i \cdot f\} - \min_{i}\{D^i \cdot f\}$, subject to $\min_{i}\{D^i \cdot f\} \ge P \cdot f$. For reasons we explain later, we guess and fix the minimizer for $\min_{i}\{D^i \cdot f\}$, denoted by $D'$.

The LP: Maximize $x - D' \cdot f$, for the variables $x$ (scalar) and $f$ (vector), subject to:
\begin{enumerate}
    \item Frequencies: (i) $\forall 1 \le i \le n: f_i \ge 0$;\ \ \ \ (ii)  $\sum_{i=1}^{n}{f_i} = 1$.
    \item Hierarchy: (iii) $P \cdot f \le D' \cdot f$;\ \ \ \ (iv) $\forall 1 \le i \le N: D' \cdot f \le D^i \cdot f$.
    \item Technical helper variable: (v) $\forall 1 \le i \le M: x \le S^i \cdot f$.
\end{enumerate}
\end{definition}

Ideally we would define two helper variables $x$ and $y$ such that the constraints give us $x = \min_{i}\{S^i \cdot f\}$ and $y = \min_{i}\{D^i \cdot f\}$, and then maximize $x-y$. We defined $x$, but cannot do so for $y$. Since we maximize $x-y$ attempting to do it by requiring $\forall 1 \le i \le N: y \le D^i \cdot f$ will simply lead to an unbounded small value for $y$ (or the minimum possible if we artificially bound it from below, say with $y \ge 0$). For this reason, we compromised and solved for each option of $D' \in \{D^1,\ldots,D^N\}$. The LP has no solution for some wrong guesses of $D'$, but this does not affect us. However, solving this large number of instances is highly inefficient, and is feasible only for small topologies. We obtained from this LP the directions marked by ${*}/{**}$ in Table~\ref{table_approximation_ratios}. Note that just like with the integrality gap, we optimized the direction $f$ to maximize a difference rather than the actual approximation ratio, so the actual lower bound may be larger.

\begin{remark}
\label{remark_missing_condition_for_LP_of_gap}
Critically, a solution $f$ to the LP in Definition~\ref{definition_LP_for_approximation_gap} is not necessarily what we look for, because in order to argue about root-rounding and the approximation ratio, $P$ must be optimal with respect to $f$. Otherwise we root-rounded the wrong point. The LP as we defined it does not guarantee it, however all the directions in Table~\ref{table_approximation_ratios} do satisfy this additional condition. Furthermore, interestingly, the plot thickens: In the rows in Table~\ref{table_approximation_ratios} that are marked with ${*}$ (first six) and the one marked with $p'$, the vertex $P$ is not a unique optimum for the chosen direction, and in fact one or more STTs are optimal too. However, the gap analysis is still valid since rounding $P$ is sub-optimal, and the fact that the optimum is a face that contains an integer solution rather than a single fractional point does not mean we can find an integer optimal point efficiently in general (on small topologies, we can just compare $P \cdot f$ to $D \cdot f$ for every $D$ induced by an STT.)
\end{remark}

\begin{remark}
\label{remark_int_gap_zero_must_perturb}
Definition~\ref{definition_LP_for_approximation_gap} technically allows for $P \cdot f = D' \cdot f$. Then there is no integrality gap, and when we solve the original LP (Definition~\ref{definition_LP_program}) we may find an STT induced vertex, so that we do not even need to approximate, in contrast to what we aim to find. This, in fact, happens when we run the solver on the topology $U_{(8,13)}$, see the direction marked in ${**}$ in Table~\ref{table_approximation_ratios}. To overcome this issue, we replace constraint (iii) by the constraint $P \cdot f + \epsilon \le D' \cdot f$ for a small $\epsilon > 0$, say $0.001$, to truly guarantee an approximation ratio that is larger than $1$.
\end{remark}

Moving on to maximizing the worst-case approximation ratio, we can try a similar LP as in Definition~\ref{definition_LP_for_approximation_gap}. In this case, using the same notation, we would like to maximize $\max_{i}\{S^i \cdot f\} - \min_{i}\{D^i \cdot f\}$. The problem here is that we cannot define the helper-variable $x$ such that modifying its constraints gives $x = \max_{i}\{S^i \cdot f\}$ while simultaneously maximizing $x - D' \cdot f$. Instead, we can guess an $S' \in \{ S^1,\ldots,S^M\}$ that maximizes the expression, similar to our guess of $D'$. The problem is that we get an additional slowdown factor of $M$, since we now solve $M \cdot N$ instances (one per guess of $D'$ and $S'$). For the small topologies we already have $M \approx 30$, depending on the vertex that we round.

At this point, it does not seem too important to pinpoint the exact approximation ratios for only a few specific small topologies. So, as a crude replacement, we may check each primary direction for its worst-case rounding. This is feasible, and we find some larger ratios compared to the best-case roundings, as should be expected. The primary directions for topologies $U_{(8,4)}$ and $U_{(8,13)}$ break the $1.3$ barrier, 
they appear in Table~\ref{table_approximation_ratios} marked with ${p}$. To emphasize that these primary directions do not maximize the approximation ratio, we also show a deviation from one of the directions, marked with ${p'}$ in the table. It is possible to deviate from a primary direction by not too much while maintaining the same vertex solution of the LP. Thus, by considering some random deviations we were able to find an improved direction. We emphasize that this is likely not the direction that maximizes the approximation ratio as our search was not exhaustive.
\section{Positive Results}
\label{section_positive_results}

In this section we prove some positive results, including the facts that the LP for topologies with $n \ge 3$ nodes always has integer vertices that are not induced by STTs, and that the LP for a star graph has only integer vertices, implying optimality of the LP for stars of any size.

\subsection{Basic Properties}
\label{subsection_basic_properties_proven}

\begin{lemma}[Depths Bound]
\label{lemma_depth_bounds}
Let $(X,Z,D)$ be a feasible solution of the LP. Then for \emph{every} subset of coordinates $S \subseteq \{1,\ldots,n\}$ it holds that: $\sum_{i \in S}{D_i} \ge |S|-1$.
\end{lemma}
The claim is obvious for points induced by STTs, since the LP-depth of every node except for the root is at least one. For the general case, we argue as follows.

\begin{proof}
Fix $i$ and $j$ and consider their ancestry constraint:
$
1
\le X_{ij} + X_{ji} + \sum_{k \in (i \leftrightsquigarrow j)}{Z_{kij}}
\le X_{ij} + X_{ji} + \sum_{k \in (i \leftrightsquigarrow j)}{\min\{X_{ki},X_{kj}\}}
\le X_{ij} + X_{ji} + \frac{1}{2} \sum_{k \in (i \leftrightsquigarrow j)}{(X_{ki}+X_{kj})}
\le \frac{1}{2} (X_{ij} + X_{ji} + D_i + D_j)
$. Then: $\binom{|S|}{2} \le \sum_{\{i,j\} \subset S}{\frac{1}{2} (X_{ij} + X_{ji} + D_i + D_j)}$. Note that each $D_i$ appears exactly $|S|-1$ times on the right side, and that we have one more contribution of (at most) $D_i$ from the variables $X_{ji}$. Therefore: 
$\binom{|S|}{2} \le \frac{|S|}{2} \sum_{i \in S}{D_i} \Rightarrow |S|-1 \le \sum_{i \in S}{D_i}$.
\end{proof}

\begin{corollary}
\label{corollary_upper_bound_small_depth_counts}
The depths-vector of any feasible solution has at most $k$ entries with value less than $\frac{k}{k+1}$. In particular, at most one coordinate smaller than $\frac{1}{2}$, at most two smaller than $\frac{2}{3}$, etc.
\end{corollary}

\begin{remark} 
Corollary~\ref{corollary_upper_bound_small_depth_counts} interprets Lemma~\ref{lemma_depth_bounds} as an upper-bound on the number of nodes with small depth $d < 1$ (it is not meaningful for $d \ge 1$). However, there are no lower-bound guarantees. Unlike STT induced vertices that are guaranteed to have one $0$ depth (the root's), a general vertex could have only ``large'' coordinates.

It may be interesting to understand whether the minimum depth (coordinate) of a vertex is bounded from above. We emphasize that this question is interesting only with respect to vertices rather than with respect to general feasible solutions, because it is easy to take a combination of STTs to get a feasible point whose minimum depth is large. As a concrete example of vertices whose minimum depth is $1$, there are four such vertices up to automorphisms for the topology $U_{(8,4)}$ (Figure~\ref{figure_all_trees_up_to_n8}) which are $(2, 2, 4.5, 1.5, 1, 1, 2, 2)$, $(1.5, 2, 4.5, 1.5, 1, 1, 2.5, 2.5)$, $(2.5, 2.5, 4.5, 1.5, 1.5, 1.5, 1.5, 1)$, $(1.5, 2, 5, 1.5, 1, 1, 2.5, 2.5)$ and one vertex for the topology $U_{(8,13)}$ which is $(1.5, 2, 3.5, 2, 1, 1, 3.5, 1.5)$.
\end{remark}

Lemma~\ref{lemma_depth_bounds} is agnostic to weights. The following Lemma~\ref{lemma_node_depth_by_weight} takes weights into account.

\begin{lemma}[Depth by weights]
\label{lemma_node_depth_by_weight}
Let $i$ be a node whose frequency is $f_i$ where $\sum_{u \in U}{f_u} = 1$. Then $depth(i) \le \frac{1}{f_i}$ in every optimal STT. (Recall that the root's depth is $1$; interpret $\frac{1}{0} = \infty$.)
\end{lemma}
\begin{proof}
Consider an STT $T_1$ where $i$ has depth $k > \frac{1}{f_i}$. Define $T_2$ by promoting $i$ to the root. Each node other than $i$ gains at most one new ancestor ($i$), so we lose a cost of at most $\sum_{j (\ne i) \in U}{f_j} = 1 - f_i$. However, we also save a cost of $(k-1) \cdot f_i$ for $i$. In total: $cost(T_1) - cost(T_2) \ge (k-1) f_i - (1- f_i) = k \cdot f_i - 1 > 0$. So $T_1$ is sub-optimal.
\end{proof}

\begin{corollary}
If there is a node $i$ with frequency $f_i > \frac{1}{2}$, it must be the root of every optimal STT. Therefore, we can focus our efforts on inputs where $\forall i: f_i \le \frac{1}{2}$. 
\end{corollary}

\subsection{More Analysis of Integer Vertices}
\label{subscetion_note_on_STT_as_vertices}
In this section we motivate focusing our interest on the lower envelope of the $D$-space polytope as we discussed in Section~\ref{subsection_general_counter_exaple} for the normals method.

\begin{theorem}
\label{theorem_integer_vertex_iff_stt}
Let $\mathbb{P}$ be the polytope of the LP corresponding to a topology $U$, and let $P \in \mathbb{P}$ be an integer point. If $P$ is induced by some STT $T$ then $P$ is a vertex on the lower envelope of $\mathbb{P}$. Moreover, if the projection of $P$ to $D$-space is a vertex on the lower envelope of the projection of $\mathbb{P}$, then there exists an STT $T$ with depths vector $P_D$.
\end{theorem}

Note that Theorem~\ref{theorem_integer_vertex_iff_stt} is almost an ``if and only if'' statement, except that the first part of the claim considers the $(X,Z,D)$-polytope while the second part considers the projection to $D$-space.

\begin{proof}
We distinguish projections to $D$-space from the corresponding objects in the $(X,Z,D)$-polytope by a subscript $D$. Consider the second part of the claim, that $P_D$ is a vertex on the lower-envelope of $\mathbb{P}_D$, and assume by contradiction that there is no STT $T$ with depths vector $P_D$.
Then by being a vertex of $\mathbb{P}_D$, there is a direction $w$ for which $P_D$ is the unique optimum of the LP. Since $P_D$ is on the lower envelope of $\mathbb{P}_D$, there is such $w$ with non-negative coordinates. This $w$ is a natural objective function to Problem~\ref{problem_main_problem_stt}, and from Property~\ref{property_integer_domination} we deduce that $P_D$ is not the unique optimum in the direction $w$, a contradiction. Note that our argument did not show that $P_{X,Z}$ is also induced by the $T$.

For the first part of the claim, define the weight of each node $i$ as $w_i = n^{-4 \cdot depth_T(i)}$ (recall that $T$ induces $P$). We next show that $P$ is the unique optimal solution for the LP in the direction $w$. This implies that $P$ is a vertex, and since all the weights are positive, that $P$ is on the lower envelope of $\mathbb{P}$.

To show that $P$ is uniquely optimal in the direction $w$, we prove that every optimal solution satisfies $\forall i: X_{ir}=0$ (this holds for $P$). Then $X_{ir}=0 \Rightarrow Z_{ijr} = 0 \; (\le X_{ir})$ and we deduce from the ancestry constraint on $r$ and $j$ that $\forall j: X_{rj} = 1$. The argument is concluded by recursively considering the connected components of the topology when $r$ is removed, and eventually after we fix all the variables, the set of optimal solutions turns out to be a single point, $P$.

Consider a feasible solution with $X$ coordinates such that there is some $i \ne r$ such that $X_{ir} > 0$, and let $\Delta \equiv \min \{ X_{ir} \mid X_{ir} > 0 \}$. We can infer values of  the other variables by setting: $D_i = \sum_{j}{X_{ji}}$ and $Z_{kij} = \min\{X_{ki},X_{kj}\}$.\footnote{There could be other ways to set $Z$ (smaller) and $D$ (larger), but our choice is feasible. Importantly, we cannot choose smaller values of $D$.} We define an alternative feasible solution $X'$ where $X'_{ij} = X_{ij} + n\Delta$ for $j \ne r$, $X'_{ir} = \max\{0,X_{ir} - \Delta\}$ and $Z'_{kij} = \min\{ X'_{ki},X'_{kj} \}$. If $j=r$ then $Z'_{kir} \ge Z_{kir} - \Delta$ ($Z'_{kir} = Z_{kir}$ if $X_{kr} = 0$), and if $r \ne i,j$ then $Z'_{kij} = Z_{kij} + n\Delta$.

By construction we maintain the non-negativity of $X'$ (and hence of $Z'$). Ancestry constraints are preserved because all the variables increase except for those of the form $X_{ir}$ and $Z_{kir}$ ($i \ne r$). $X_{ir}$ and $Z_{kir}$ only decrease if $X_{ir} > 0$, and can only violate the ancestry constraints over the pair $r$ and $i$. However, this constraint is not violated since: $X'_{ri} + X'_{ir} + \sum_{k \in (i \leftrightsquigarrow r)}{Z'_{kir}} \ge 
(X_{ri}+n\Delta) + (X_{ir}-\Delta) + \sum_{k \in (i \leftrightsquigarrow r)}{(Z_{kir} - \Delta)} \ge 1 + \Delta$. Finally, observe that $D'_r \le D_r - \Delta$ and for $i \ne r$: $D'_i \le D_i + n \cdot n\Delta$. Therefore the difference in values of the LP is:
$value(X',w) - value(X,w) = \sum_{i}{w_i \cdot (D'_i - D_i)} \le \sum_{i \ne r}{\frac{w_i}{w_r} \cdot w_r \cdot n^2 \Delta} - w_r \Delta < (n \cdot \frac{1}{n^4} \cdot n^2 - 1) \cdot w_r \Delta  < 0$. Therefore, every optimal point with respect to $w$ satisfies that $\forall i: X_{ir} = 0$.
\end{proof}

Theorem~\ref{theorem_integer_vertex_iff_stt} does not rule out an integer vertex $P$ of the LP such that $P_D$ is not on the lower envelope of $\mathbb{P}_D$. Indeed, Theorem~\ref{theorem_non_tree_vertex} shows that such vertices exist. Conveniently, their projections do not affect our normals method since we relax $\mathbb{P}_D$ to only deal with lower-envelope facets.

\begin{theorem}
\label{theorem_non_tree_vertex}
For any $n \ge 3$, and any tree topology $U$ with $n$ nodes, the LP defined by $U$ has integer vertices that are not induced by an STT (Definition~\ref{definition_xzd_of_stt}).
\end{theorem}

The proof shows two generic classes of such integer vertices. The first relies on non-transitivity in ancestral relationship, in opposition to an STT property not captured by the LP. The second defines an LCA implicitly but not explicitly.

\begin{proof}
Given a topology $U$, let $\mathbb{P}'$ be the set of all feasible solutions such that every $X$ or $Z$ variable is $0$ or $1$, and in each ancestry constraint, exactly one variable is $1$. $D$ variables are non-binary and always satisfy $D_i = \sum_{j} {X_{ji}}$, thus the following argument addresses only the $X,Z$ variables. We argue that every point in $\mathbb{P}'$ is a vertex. Let $P \in \mathbb{P}'$.
$P$ is a vertex if and only if there is no non-zero direction $R$ such that $P \pm \epsilon R$ are both feasible. If $R$ is nonzero in some zero coordinate of $P$, one direction is infeasible due to negativity. Similarly, $R$ cannot be nonzero in  a coordinate where $P$ is $1$, because in order to preserve the ancestry feasibility to which this coordinate is the sole  contributor, $R$ must be non-zero on some other zero coordinate of $P$ with opposite sign. But since we already argued that $R$ is zero in all the zero coordinates of $P$, no such direction $R$ exists and we conclude that $P$ is a vertex.

Let $\mathbb{P} \subseteq \mathbb{P}'$ contain all the feasible solutions such that
$Z_{kij}=0$ for all $i$, $j$, and $k\in (i \leftrightsquigarrow j)$, and for every pair $i,j$ choose $X_{ij},X_{ji} \in \{0,1\}$ such that $X_{ij} + X_{ji} = 1$. We define two classes of points in $\mathbb{P}$ which do not correspond to STTs.

\begin{enumerate}
    \item We choose $P \in \mathbb{P}$ such that there is a triplet of indices $i,j,k$ (this is why the claim requires $n \ge 3$) for which $X_{ij} = X_{jk} = X_{ki} = 1$. No STT has such a cyclic relation, because if $i$ is an ancestor of $j$ which is an ancestor of $k$, then $i$ is also an ancestor of $k$.

    \item Start from a point $P' \in \mathbb{P}'$ induced by an STT $T$ that has at least one node that is an LCA ($T$ is not a path). Denote the LCA by $a$ and the two nodes that it separates by $b$ and $c$. Define $P$ to have the same values as $P'$ except that we set $Z_{abc} = 0$ and $X_{bc}=1,X_{cb}=0$. One can verify that $P$ is feasible.
    We say that $a$ remains an implicit ancestor of $b$ and $c$ since we still have that $X_{ab}=1$, $X_{ac}=1$, but 
    $Z_{abc} = 0$.
    If we replace every variable $Z_{kij}=1$ as we did for $Z_{abc}$, we  get a point in $\mathbb{P}$. Note that as a side-effect, we showed some non-STT integer vertices that belong to $\mathbb{P}'$ and not to $\mathbb{P}$ (some ancestors were made implicit but not all of them). \qedhere
\end{enumerate}
\end{proof}

\begin{remark}
\label{remark_most_integer_vertices_are_non_stt}
Theorem~\ref{theorem_non_tree_vertex} shows the existence of non-STT, even integer, vertices. In fact, a simple counting argument shows that they are the majority for large $n$.
Let $S$ be the set of all STT vertices, and $\mathbb{P}$ and $\mathbb{P}'$ as in the proof of Theorem \ref{theorem_non_tree_vertex}. We have that $|\mathbb{P}'| > |\mathbb{P}| = 2^{\binom{n}{2}}$, while $|S| \le n!$ since a permutation uniquely defines an STT by choosing a root and recursing on the suffix (multiple permutations may define the same STT). Finally, note that $S \subset \mathbb{P}'$.
\end{remark}

To conclude this section, we analyze all the vertices of the polytopes for topologies with $n=2$ and $n=3$ nodes, each of these sizes has a single, path, topology. In both cases, there are only integer vertices. When $n=2$, there are no $Z$ variables, the LP is defined as $X_{12}+X_{21} \ge 1$ for $X_{12},X_{21} \ge 0$. Overall there are two vertices, one per STT. When $n=3$ the LP becomes more interesting. The following inequalities define the polytope, and Table~\ref{table_vertices_n3} lists all $9$ vertices: $5$ of them correspond to STTs, and $4$ more are constructed as detailed in the proof of Theorem~\ref{theorem_non_tree_vertex}, two of each class.
\begin{equation*}
    X_{12}+X_{21} \ge 1, \ \ \ 
    X_{23}+X_{32} \ge 1,  \ \ \ 
    X_{13}+X_{31}+Z_{213} \ge 1,  \ \ \ 
    Z_{213} \le X_{21},  \ \ \ 
    Z_{213} \le X_{23}  \ \ \ 
\end{equation*}

\begin{table}[!ht]%[!ht]
    \scriptsize
    \begin{center}
    \begin{tabular}{|c|c|c|c|c|c|c|c|c|l|}
    \hline
    $X_{12}$ & $X_{21}$ & $X_{23}$ & $X_{32}$ & $X_{13}$ & $X_{31}$ & $Z_{213}$ & Comments \\
    \hline
    \hline
    0 & 1 & 0 & 1 & 0 & 1 & 0 & STT: root $3$, leaf $1$ \\
    \hline
    0 & 1 & 0 & 1 & 1 & 0 & 0 & non-STT: cyclic ancestry  \\
    \hline
    0 & 1 & 1 & 0 & 0 & 0 & 1 & STT: root $2$ \\
    \hline
    0 & 1 & 1 & 0 & 0 & 1 & 0 & non-STT: LCA abuse  \\
    \hline
    0 & 1 & 1 & 0 & 1 & 0 & 0 & non-STT: LCA abuse  \\
    \hline
    1 & 0 & 0 & 1 & 0 & 1 & 0 & STT: root $3$, leaf $2$ \\
    \hline
    1 & 0 & 0 & 1 & 1 & 0 & 0 & STT: root $1$, leaf $2$  \\
    \hline
    1 & 0 & 1 & 0 & 1 & 0 & 0 & STT: root $1$, leaf $3$  \\
    \hline
    1 & 0 & 1 & 0 & 0 & 1 & 0 & non-STT: cyclic ancestry \\
    \hline
    \hline
    \end{tabular}
    \end{center}
    \caption{\footnotesize{All the vertices of the LP of the topology over $n=3$ nodes. The vertices with comment of ``non-STT: LCA abuse'' become non-vertices when we eliminate $Z_{213}$ as detailed in Section~\ref{subsection_with_and_without_z_fourier_motzkin}.}}
    \label{table_vertices_n3}
\end{table}

\subsection{Star Topologies Have Integer Vertices}
\label{subsection_star_is_integer}

In this subsection we prove analytically that if the tree topology is a star, then all the vertices of the LP are integer. Section~\ref{subscetion_note_on_STT_as_vertices} explicitly shows this claim for $n=2,3$. Consequently, as a result of Property~\ref{property_integer_domination}, we get a polynomial-time algorithm to find the optimal static STT (for stars). 

\begin{theorem}
\label{theorem_star_has_integer_vertices}
The LP polytope for a star topology has only integer vertices.
\end{theorem}

\begin{corollary}
An optimal static STT for any star topology is computable in polynomial time.
\end{corollary}

(The corollary can be proven directly, using a simple enumeration: since the star is very symmetric, there are at most $n$ effective STTs to consider: we always query leaves in decreasing order of weight, and the $n$ STTs vary according to the depth at which we query the central node.)

\begin{proof}[Proof of Theorem~\ref{theorem_star_has_integer_vertices}]
The LP for a star over $n$ nodes is given in Equation~(\ref{equation_starn_original}), where we denote the center by $1$ (recall that $Z_{1ij}$ and $Z_{1ji}$ are the same).
\begin{equation}
\label{equation_starn_original}
    \forall i>1: X_{1i}+X_{i1} \ge 1 ; \\  \ \ \ \ \ 
    \forall i,j \ne 1: X_{ij} + X_{ji} + Z_{1ij} \ge 1 ; \\  \ \ \ \ \ 
    \forall i,j \ne 1: Z_{1ij} \le X_{1i} ; \\  \ \ \ \ \ 
    X,Z \ge 0 \\
\end{equation}

We only regard the $X$ and $Z$ variables, $D$ is set tightly according to $D_i = \sum_{j \ne i}{X_{ji}}$. We show that each vertex corresponds to an assignment satisfying the following three conditions:
\begin{enumerate}
    \item \label{assign1} For every $i,j \ne 1$, $Z_{1ij} \in \{0,1\}$.
    \item \label{assign2} For each $i \ne 1$: Let $z_i \equiv \max_{j \ne 1,i} \{Z_{1ij}\}$. We set $X_{1i} \in \{z_i,1\}$ and $X_{i1} = 1 - X_{1i}$.
    \item \label{assign3} For each $i,j \ne 1$: Choose $X_{ji} \in \{0,1-Z_{1ij}\}$ and set $X_{ij}=1-Z_{1ij}-X_{ji}$. Explicitly: if $Z_{1ij}=1$ then $X_{ij}=X_{ji}=0$, otherwise ($Z_{1ij} = 0$) $X_{ij} \in \{0,1\}$ and $X_{ji} = 1-X_{ij}$.
\end{enumerate}
One can verify that each such assignment defines a feasible integer point. We proceed to prove that every vertex corresponds to an assignment that satisfies these conditions. We do this by showing that if any of the three conditions does not hold at a feasible point $P$, then there is a direction $R$ such that $P'_{\pm} = P \pm \epsilon R$ are both feasible for some $\epsilon > 0$, thus $P = \frac{1}{2} (P'_{+} + P'_{-})$ is not a vertex. The coordinates of $R$ are either $0$ or $\pm 1$, and we say that a coordinate of $R$ 
that equals $1$ is \emph{increased}
and a coordinate of $R$ that equals $-1$ is \emph{decreased}.
 Once we prove that some condition holds, we assume it to hold when we consider other conditions in later arguments.

First note that every variable is bounded in $[0,1]$. Non-negativity is by definition of the LP. If there are variables larger than $1$, let $R$ decrease all the variables that are larger than $1$ (indeed $P \pm \epsilon R$ are both feasible).

To prove Condition~\ref{assign3}, for fixed $i,j \ne 1$ note that these variables only participate in the constraint $X_{ij}+X_{ji} \ge 1 - Z_{1ij}$. If $X_{ij}+X_{ji} > 1 - Z_{1ij}$, since $Z_{1ij} \le 1$ one of $X_{ij}$ and $X_{ji}$ is positive, and we can choose $R$ to decrease it. Thus $X_{ij} + X_{ji} = 1 - Z_{1ij}$. If $X_{ij},X_{ji} \in (0,1-Z_{1ij})$ (note that if $Z_{1ij}=1$ this range is empty) then both are positive, and $R$ increases $X_{ij}$ and decreases $X_{ji}$. We conclude that $X_{ij}$ is either $0$ or $1-Z_{1ij}$ and $X_{ji} = 1-Z_{1ij}-X_{ij}$. This proves Condition~\ref{assign3}.

To prove Condition~\ref{assign2}, consider the case $X_{i1} + X_{1i} > 1$ for some $i$. Since both variables are at most $1$, both are positive, and $X_{i1}$ only participates in this single inequality, so take $R$ which increases $X_{i1}$. Now that we have $X_{i1} + X_{1i} = 1$, if $z_i < X_{1i} < 1$ then $X_{i1} > 0$, so take $R$ which decreases $X_{1i}$ and increases $X_{i1}$. (In both cases, the strict inequality ensures that $P \pm \epsilon R$ are both feasible, so we must have equalities as argued.)

To prove Condition~\ref{assign1}, let $0<a<1$ be the largest non-integer value of some $Z$ variable, if such a value exists. Let $E_a$ be the set of unordered leaf pairs $\{i,j\}$ for which $Z_{1ij}=a$. For each $\{i,j\} \in E_a$ we have $(X_{ij}+X_{ji}) \ge (1-a) > 0$, so let $A_{ij}$ denote a positive variable among $X_{ij}$ and $X_{ji}$ (if both are, pick one arbitrarily). Define $R(Z_{1ij})=1$ and $R(A_{ij})=-1$. For every leaf $i$ in a pair in $E_a$ such that $X_{1i}=a$, also define $R(X_{1i})=1$ and $R(X_{i1})=-1$. Set every other $X$ and $Z$ coordinate of $R$ to zero. We verify that $P \pm \epsilon R$ are both feasible for a sufficiently small $\epsilon>0$. For each pair $\{i,j\} \in E_a$, the ancestry equality is preserved because the increase of $Z_{1ij}$ is cancelled by the decrease of $A_{ij}$. If the capping constraint $Z_{1ij} \le X_{1i}$ is tight, then $X_{1i}=a$ and both sides move together; otherwise $X_{1i}=1$ by Condition~\ref{assign2} and the constraint has positive slack. The same holds for endpoint $j$. For a changed $X_{1i}$, the equality $X_{1i}+X_{i1}=1$ is preserved (and also $0 < X_{1i} , X_{i1} < 1$ so there is slack to change both variables). Any capping constraint that involves this $X_{1i}$ against $Z_{1ij}$ for $\{i,j\} \notin E_a$ has slack since by maximality and the definition of $E_a$, $Z_{1ij} < a$. Finally, all coordinates that are decreased are positive, so non-negativity is preserved for sufficiently small $\epsilon$.
\end{proof}

\subsection{Extending Topologies}
\label{subsection_graph_extensions}

In this section we prove that if the LP for a topology has non-integer vertices, the LP of topologies that extend it in some ways also have non-integer vertices. This explains formally, for example, why the three topologies with $8$ nodes that are colored in red in Figure~\ref{figure_all_trees_up_to_n8} must have non-integer vertices once we know that the red topology with $7$ nodes has non-integer vertices. We then discuss another kind of extension, when we combine topologies through a new common node.

\begin{definition}
\label{definition_extended_graph}
We define a single extension of a graph $G$ to $G'$ and denote it $G \to G'$ as the action of adding a new node to $G$ and connecting it to some existing node as a new leaf, or between two adjacent nodes, i.e., subdividing an existing edge into two new edges.
We say that $G$ extends $H$, or that $H$ is a shrinking of $G$, if there is a sequence of extensions of length $|G|-|H|$ such that $H \to H' \to \ldots \to G$.
\end{definition}

\begin{theorem}
\label{theorem_subgraph_with_frac_vertex_implies_frac_vertex_extension}
If the LP for a topology $U$ has non-integer vertices, so does the LP for any extension of $U$. The claim holds both for the original space of $(X,Z,D)$ variables, and also when only considering the depths $D$.
\end{theorem}

\begin{proof}
It is sufficient to prove the claim for a single extension from $U$ to $U'$. Then if $U^*$ is an extension of $U$, by applying the claim over a single extension $|U^*|-|U|$ times, we may conclude that $U^*$ too has non-integer vertices.

Throughout the proof we  use prime to denote terms that correspond to $U'$. For example, we denote the set of variables in the LP of $U$ by $(X,Z,D)$, and  $(X',Z',D')$ denote the variables in the LP of $U'$. Since $U'$ is a single extension of $U$, $(X',Z',D')$ is a union of $(X,Z,D)$ with additional new variables that correspond to the new extending node, which we denote by $a$. 

Let $P$ be a non-integral vertex of the LP for $U$.  We first extend $P$ to a feasible point $P'$ of the LP for $U'$ while preserving all old coordinates.  After giving this construction, we show that an appropriate contraction of $U'$ back to $U$ maps some vertex of the extended LP to $P$.  This will imply that the extended LP also has a non-integral vertex.

We explain how to extend $P$ to $P'$ that is a feasible point for the LP of $U'$, not necessarily a vertex, by setting the values of the extra variables. Loosely speaking, we set the values so that $a$ is a ``descendant'' of its neighbor(s). The values of $X,Z,D$ are the same in $P'$ and $P$, and we only need to define the values of $X'\setminus X$, $Z' \setminus Z$ and $D_a (= D'\setminus D)$. Because $U'$ is a single extension, we have two cases to consider. The first and simpler case is when the added node $a$ is a leaf, denote its neighbor by $b \in U$, and $U_b \equiv U \setminus \{b\}$. Then the extension is as follows:
\begin{enumerate}
    \item $X' \setminus X$: Set $X'_{ba}=1$, $\forall i \in U_b: X'_{ia} = X_{ib}$ and $\forall i \in U: X'_{ai} = 0$ (child of $b$, ancestor of none).
    \item $Z' \setminus Z$: Set $Z'_{kai} = Z_{kbi}$ for $i,k \in U_b$ and $Z'_{bai} \equiv \min\{ X'_{ba}, X'_{bi} \} = X_{bi}$. Since $a$ is a leaf in the topology there are no variables of the form $Z'_{aij}$.
    \item Set $D'_a = \sum_{i \in U}{X'_{ia}}$.
\end{enumerate}

To verify feasibility, note that any constraint involving only old variables is satisfied by construction because we copied the values. So we only need to check constraints involving variables associated with $a$. By construction $D'_a \ge \sum_{i \in U}{X'_{ia}}$ (in fact, equal), and since $\forall i \ne a: X'_{ai}=0$ we have indeed $D'_i = D_i$. Next, $Z'_{iaj} \le \min\{ X'_{ia},X'_{ij} \}$ is explicitly stated if $i=b$, and for $i \ne b$: $Z'_{iaj} = Z_{ibj} \le \min\{ X_{ib},X_{ij} \} = \min\{ X'_{ia},X'_{ij} \}$. Finally,
consider ancestry constraints for $a$ and some $i \in U$. If $i=b$: Then $X'_{ba} = 1$ and we are done. Otherwise, $X'_{ai} + X'_{ia} + \sum_{k \in (a \leftrightsquigarrow i)}{Z'_{kai}} = 0 + X_{ib} + (X_{bi} + \sum_{k \in (b \leftrightsquigarrow i)}{Z_{kbi}}) \ge 1$ where the last inequality is known to hold, in $U$, with respect to the pair $i,b$.

In the second case node $a$ subdivides the edge $(b,c)$ between two nodes $b,c \in U$. Define $U'_x$ for $x \in \{b,c\}$ to be the connected component of $U'$ that contains $x$ when we delete $a$. We set the $X$ and $Z$ variables that involve $a$ and nodes in $U'_x$ as if $a$ is a leaf that was added to $U'_x$, and set $D'_a = \sum_{i \in U}{X'_{ia}}$. It only remains to define values for the variables  $Z'_{aij}$ for $i \in U'_b$ and $j \in U'_c$ (a case we did not have earlier), and we trivially set them to $0$.

To verify feasibility, note that setting $Z'_{aij} = 0$ for $i \in U'_b$ and $j \in U'_c$ is fine because $Z'_{aij} = 0 \le X'_{ai} = X'_{aj}  = 0$. Then the only non-trivial constraints that require verification are ancestry constraints. However, for pairs $i,j \in U$ we have them satisfied by construction (copied values), and for pairs of $a$ with $i \in U'_x$ we know them to be satisfied because we verified this relation in the previous case of $a$ being an actual leaf (here $U'_y$ for $y \ne x$ has no effect on the constraints).

It remains to show that the preceding construction preserves the
existence of non-integral vertices.  For a leaf addition, let the map $\Phi_1$
simply discard every coordinate involving the new leaf. Constraints
on old nodes are unchanged, so $\Phi_1$ maps every feasible point for
$U'$ to a feasible point for $U$; moreover, the construction above
shows that $\Phi_1$ is surjective. In the case where $a$ subdivides the edge $(b,c)$, directly discarding the coordinates involving $a$ might not preserve feasibility, because $Z'_{aij}$ may contribute to an ancestry constraint. Instead, contract $a$ into $b$. Formally, for a feasible point $(X',Z',D')$, define
$(\hat X,\hat Z,\hat D)=\Phi_2(X',Z',D')$ by:
\begin{enumerate}
    \item For every $i \in U$: $\hat D_i=D'_i$.
    
    \item For $i, j \in U$, $i \ne j$: If $i=b$ then $\hat X_{bj}=X'_{bj}+X'_{aj}$, otherwise $\hat X_{ij}=X'_{ij}$.

    \item For $k\in(i\leftrightsquigarrow_U j)$: If $k=b$ and $a\in(i\leftrightsquigarrow_{U'}j)$ then $\hat Z_{bij} = Z'_{bij}+Z'_{aij}$, otherwise $\hat Z_{kij} = Z'_{kij}$.
\end{enumerate}

This contraction preserves feasibility.  Indeed, whenever two
$Z$-variables are combined their capping constraints give
$\hat Z_{bij}=Z'_{bij}+Z'_{aij} \le X'_{bi}+X'_{ai}=\hat X_{bi}$, and similarly we get $\hat Z_{bij} \le \hat X_{bj}$. Other capping constraints are unaffected since $\hat Z_{kij} = Z'_{kij}$ while $\hat X \ge X'$ coordinate-wise. In an ancestry constraint whose path uses $a$, the contribution $Z'_{aij}$ is either absorbed into $\hat Z_{bij}$, if $b$ is internal to the old path, or is bounded by the added term $X'_{aj}$ in $\hat X_{bj}$. All other ancestry constraints only increase or remain unchanged. Finally, for $j \ne b$ the depth sum is preserved, $\sum_{i\in U\setminus\{j\}}\hat X_{ij}
=\sum_{i\in U'\setminus\{j\}}X'_{ij}\le D'_j=\hat D_j$,
whereas for $j=b$ it can only decrease, because the non-negative term
$X'_{ab}$ is discarded. Thus $\Phi_2$ maps the feasible region for
$U'$ into that for $U$. The extension constructed above has
$X'_{ai}=0$ for every old node $i$ and $Z'_{aij}=0$ whenever $a \in (i \leftrightsquigarrow_{U'} j)$. Hence $\Phi_2(P')=P$, and $\Phi_2$ is
surjective. In either case, write $\Phi$ for the resulting contraction.

Let $h$ be a direction uniquely minimized at $P$ in the LP for $U$, and use the pullback direction $h'=\Phi^{\mathsf T} h$ for the LP of $U'$.\footnote{$\Phi^{\mathsf T}$ is the transpose of $\Phi$ as a linear transformation which maps $h$ (smaller space) to $h'$ (larger space).} The minimum of $h'$ is attained, since the feasible extension $P'$
has value $h(P)$, and every feasible $Q'$ satisfies $h' \cdot Q' = \Phi^{\mathsf T} h \cdot Q' = h \cdot \Phi(Q') \ge h\cdot P$. Choose a vertex $Q'$ of the minimizing face.  The uniqueness of $P$ implies $\Phi(Q')=P$.  Such a vertex exists because the LP is contained in the non-negative orthant. Since $\Phi$ has integer coefficients, $Q'$ cannot be integral when $P$ is non-integral. This concludes the proof for the $(X,Z,D)$ space.

To argue for $D$-space, note that both contractions keep all old depth coordinates, and the extensions above preserve them.  Consequently, restriction to the old depth coordinates maps the depth projection for $U'$ onto that for $U$.  Applying the same $(h,h')$-directions argument shows that every non-integral vertex in depth space lifts to a non-integral vertex in the depth space of $U'$.
\end{proof}

\begin{remark}
\label{remark_subgraph_notes}
Note that the converse is of course not true: if an extension has non-integer vertices, it does not imply that the smaller topology has non-integer vertices. Indeed, the path topology with $3$ vertices only has integer vertices as shown in Table~\ref{table_vertices_n3}, but is extended by the long-star which has non-integer vertices (Theorem~\ref{theorem_non_integer_optimum_long_star}). The proof works since we can add zero coefficients when extending $h$ to $h'$.
\end{remark}

Next, we show that we can ``mix'' vertices of smaller topologies, in larger topologies.

\begin{theorem}
\label{theorem_combine_trees}
Let $b \ge 2$ be an integer, and let $H = \{ U^1,\ldots,U^b \} $ be a collection of topologies. Let $U$ be a new (tree) topology 
obtained by taking the union over $H$ and adding an additional new node $r$, connected to one node in each $U^i$. Then for every set of vertices $P^1,\ldots,P^b$ such that $P^i$ is a vertex of the LP for $U^i$, there is a vertex $P$ whose restriction to the internal $X,Z$ variables of $U^i$ equals that of $P^i$, and whose depths satisfy $D_u^P=D_u^{P^i}+1$ for every $u \in U^i$.
In loose terms, $P$ is a sort of Cartesian product.
\end{theorem}

\begin{proof}
In what follows, we only refer to $X$ and $Z$ variables; $D$ variables are set to $D_i = \sum_j{X_{ji}}$. We would have $D_r=0$ and $D_u=D_u^{P^i}+1$ for $u \in U^i$.

To define $P$, we intuitively set $r$ as the ``ancestor'' of every other node, which formally translates to the following values: $\forall u \in U' \setminus \{r\}: X_{ru} = 1, X_{ur} = 0$, and $\forall u \in U^i, v \in U^j, i \ne j: X_{uv}=X_{vu}=0$. Also, $D_r = \sum_{i \in U' \setminus\{r\}}{X_{ir}} = 0$. In addition, the LCA relations are as follows: $\forall u \in U' \setminus \{r\}: \forall k \in (u \leftrightsquigarrow r): Z_{kur}  = 0$ and $\forall u \in U^i, v \in U^j, i \ne j: \forall k \in (u \leftrightsquigarrow v): Z_{ruv}  = \delta_{rk}$ (i.e., $1$ iff $k=r$). So far we set all the new variables, related to $r$ or to pairs of nodes from different $U^i$s. Within each $U^i$, we copy the values of the vertex $P^i$. This results in a feasible point $P$ that is a ``Cartesian product'' of $P^1,\ldots,P^b$. Feasibility is simple to verify because $r$ behaves like a root, and within each subtree we have feasibility by definition. It remains to prove that $P$ is a vertex.

Assume by contradiction that $P$ is not a vertex, then there is a direction vector $R$ such that $P \pm \epsilon R$ is feasible. However, observe that the values that we defined explicitly are all extremal, either $0$ or $1$. Because we cannot decrease below $0$, the coordinates of variables which we set to $0$ in $P$ must be $0$ also in $R$. As for the ones that are $1$, we claim that they cannot be decreased. Decreasing $X_{ru}$ must be accompanied with increasing one of $X_{ur}$ or some $Z_{kur}$ where $k \in (u \leftrightsquigarrow r)$, but these are zeros and as we argued must be $0$ also in $R$. Similarly, decreasing some $Z_{ruv} = 1$ necessitates increasing some among $X_{uv}$, $X_{vu}$, or $Z_{kuv}$ for $k \ne r$, but again these are
$0$ in $P$ and must be $0$ in $R$. Finally, the rest of the coordinates of $R$ are partitioned into $b$ groups by the topology $U^i$ to which their corresponding vertices belong. Since $P$ restricted to coordinates of topology $j$ is a vertex in $U^j$ we conclude that $R = \vec{0}$ and that $P$ is a vertex. 
\end{proof}

\begin{remark}
\label{remark_combine_trees}
Note that Theorem~\ref{theorem_combine_trees} does not argue that every vertex of $U$ is such a ``Cartesian product'', and it is not true in general. As a concrete example, consider $U^1=U^2=U^3$ to be trees with $2$ nodes each. They only have integer vertices so every combination is integer. However, the long-star which combines them via its center node, has non-integer vertices (Theorem~\ref{theorem_non_integer_optimum_long_star}).
\end{remark}

%%% The following example has been commented out for the sake of space, even though it gives a slightly more concrete intuition. The example shows that nodes with both thirds and halves coordinates exist.
% \begin{example}
% \label{example_mixed_halves_and_thirds}
% We show in Section~\ref{subsection_with_and_without_z_fourier_motzkin} that the path topology with $n=5$ nodes has integer, half-integer, and third-integer vertices in $(X,Z,D)$-space. Then the path with $n=11$ nodes also has vertices that are mixed with halves and thirds. Indeed, when rooting the sixth node (by choosing its frequency to be very large compared to the rest), we get two independent sub-STTs and we can choose the rest of the frequencies such that one sub-STT contributes half-integers and the other contributes third-integers.
% \end{example}

To conclude this section, we close a small debt from Section~\ref{subsection_general_counter_exaple} where we mentioned that primary directions may have additional similarities that are not automorphism symmetries.

\begin{remark}
\label{remark_pseudo_symmetry}
Consider two different shrinkings of $U$ (see Definition~\ref{definition_extended_graph}), denoted by $U_1$ and $U_2$, such that they are isomorphic and each has a different  subset of the nodes of the original $U$ ($U_1 \ne U_2$). Denote the isomorphism between $U_1$ and $U_2$ by $\pi$. By Theorem~\ref{theorem_subgraph_with_frac_vertex_implies_frac_vertex_extension}, for every vertex $P_i$ of $U_i$ ($i=1,2$) we have a vertex $P'_i$ of $U$ such that the coordinates of the variables that correspond to $U_i$
in $P_i$ and $P'_i$ are the same. 
Furthermore, the proof shows how to extend a direction $h_i$ for which $P_i$ is uniquely optimal, to a direction $h'_i$ for which $P'_i$ is optimal. The extension adds $0$'s in coordinates of $h'_i$ that do not correspond to coordinates of $h_i$. So if $P_1$ is a vertex of the LP of $U_1$ and $h_1$ is a direction in which $P_1$ is the unique optimum, then $P_2 = \pi(P_1)$ is a vertex of $U_2$ and $h_2 = \pi(h_1)$ is a direction in which $P_2$ is uniquely optimal.
It follows that $P_1'$ and $P_2'$ are both vertices of $U$, $P_1'$ is optimal in the direction $h'_1$, and $P_2'$ is optimal in the direction $h'_2$. 
Note that
$h'_1$ and $h'_2$ are directions with the same coordinate-values, permuted. These values are the same as of $h_1$, and  $h_2$, padded with zeros. In many cases, this permutation does not correspond to an automorphism of $U$. As an example, name the edges of the topology $U_{(8,4)}$ as $\{(1,2),(2,3),(3,4),(4,5),(5,6),(3,7),(7,8)\}$. We can shrink it to $U_{(7,3)}$ by deleting one of the nodes in $\{4,5,6\}$, and as a result we get multiple similar but not automorphically symmetric primary directions. A concrete example of such directions is: $(3,2,0,0,2,3,3,10)$, $(3,2,0,2,0,3,3,10)$ and $(3,2,0,2,3,0,3,10)$ (all extend $(3,2,0,2,3,3,10)$ from Theorem~\ref{theorem_non_integer_optimum_long_star}). Such directions may find vertices that are not automorphically symmetric, and thus lead to a different behavior under rounding (e.g. with Definition~\ref{definition_rounding_scheme}).
\end{remark}
\section{The Dual LP}
\label{section_dual_LP}
In this section we study the dual formulation of the LP given in Definition~\ref{definition_LP_program}. Section~\ref{subsection_developing_the_dual_LP} develops the formulation slowly with verbose explanations. For the final formulation and some interpretations, skip to Definition~\ref{definition_dual_LP_program} and Section~\ref{subsection_interpreting_the_dual_LP}.

\subsection{Developing the Dual LP}
\label{subsection_developing_the_dual_LP}
First, we rewrite the primal LP:
\begin{enumerate}
    \item Remove $D$, and also change the problem to a maximization problem. As a consequence, the objective function becomes: maximize $\sum_{j \in U}{(-f_j) \cdot D_j} = \sum_{j \in U}{(-f_j) \cdot (\sum_{i \in U}{X_{ij}})}$.
    \item Rewrite the LCA inequalities as follows (for relevant values of $i,j,k$):\\ $(1)\  0 \le X_{ki} - Z_{kij} \ \ \ ; \ \ \ (2)\  0 \le X_{kj} - Z_{kij}$.
\end{enumerate}

Next, we define the variables of the dual. We have a variable $r_{ij}$ for each ancestry inequality over a pair of $i$ and $j$. There is no difference in the order of indices, so $r_{ji} \equiv r_{ij}$ and we only use one of them (similar to how $Z_{kij} \equiv Z_{kji}$). For the LCA inequalities we define the variables $q_{ikj}$ whose indices are ordered, and represent a directed ``path'' from $i$, through $k$, to $j$. For a triplet $i,k,j$, $q_{ikj}$ corresponds to the first LCA inequality, and $q_{jki}$ corresponds to the second inequality. Since all the inequalities are such that the  sum of the variables is greater than or equal to the free coefficient, both $r$ and $q$ are non-positive coordinate-wise (we flip this later).

The dual objective function is to minimize $\sum_{i,j}{r_{ij}}$ since the free coefficient is $1$ only in (ancestry) inequalities that correspond to $r$, and $0$ otherwise. We derive the dual inequalities based on the primal variables $X$ and $Z$ as follows:
\begin{enumerate}
    \item $X_{ij}$ participates in the ancestry inequality that corresponds to $r_{ij}$, and also in the LCA inequalities of the form $X_{ij} - Z_{iaj}$ where $i \in (j \leftrightsquigarrow a)$,  which corresponds to $q_{jia}$, and not to be confused with $q_{aij}$ which corresponds to $X_{ia} - Z_{ija} \ge 0$. Lastly, $X_{ij}$ has a coefficient of $-f_j$ in the objective function, therefore we get the inequality: $r_{ij} + \sum_{a : i \in (j \leftrightsquigarrow a)}{q_{jia}} \ge (-f_j)$.

    \item $Z_{kij}$ participates in three inequalities: the ones that correspond to $r_{ij}$, $q_{ikj}$ and $q_{jki}$ (with negative sign in $q$). It is not part of the primal objective, therefore: $r_{ij} - q_{ikj} - q_{jki} \ge 0$.
\end{enumerate}

It is more natural to work with non-negative variables, so let us re-define $R \equiv -r$ and $Q \equiv -q$. In this case, we flip signs in the inequalities. We also need to flip the sign in the objective function, but instead we flip from minimization to maximization. The final formulation is  as follows.

\begin{definition}[STT Dual LP]
\label{definition_dual_LP_program}
    \mbox{}
    \begin{enumerate}
        \item Variables $R$ over pairs, $Q$ over collinear triplets: $\forall i \ne j: R_{ij}$, and $\forall k \in (i \leftrightsquigarrow j): Q_{ikj},Q_{jki}$.

        \item Objective function: maximize $\sum_{i<j \ (i,j \in U)}{R_{ij}}$.

        \item Constraints:
        \begin{enumerate}
            \item (Bounds) $\forall i,j \in U, i \ne j: \forall k \in (i \leftrightsquigarrow j): 0 \le R_{ij},Q_{ikj},Q_{jki}$.

            \item (Capping) $\forall i,j \in U, i \ne j: \forall k \in (i \leftrightsquigarrow j): R_{ij} \le Q_{ikj} + Q_{jki}$.

            \item (Frequency) $\forall i,j \in U, i \ne j: R_{ij} + \sum_{a : i \in (j \leftrightsquigarrow a)}{Q_{jia}} \le f_j$.
        \end{enumerate}
    \end{enumerate}
\end{definition}

\subsection{The Dual LP: Interpretations and Insights}
\label{subsection_interpreting_the_dual_LP}

We list a few observations:
\begin{enumerate}
    
    \item The solution $R = \vec{0}$, $Q = \vec{0}$ is feasible.

    \item \label{dual_rooting} In the frequency constraints, $\{ a : i \in (j \leftrightsquigarrow a) \}$ is the set of all nodes that are ``behind'' $i$ from the perspective of $j$. Another way to view it: if we root the tree (topology) $U$ in node $j$, then this is exactly the subtree of $i$, excluding $i$.

    \item If $i$ and $j$ are neighbors, then $R_{ij}$ has no capping inequality. However, it is still bounded by frequency inequalities. If $i$ is a leaf, these are $R_{ij} \le f_j$ and $R_{ij} \le f_i - \sum_{a \in U \setminus \{i,j\}}{Q_{ija}}$.

    \item We may assume that every capping inequality holds as equality. Indeed, if $R_{ij} < Q_{ikj} + Q_{jki}$ we can decrease the right hand side until we get an equality. This can only help us increase the objective function, if we happen to get a loose frequency inequality in which we can further increase some $R_{ij}$. Given this tightness, we also get that $Q_{ik'j}+Q_{jk'i} = Q_{ikj}+Q_{jki} (= R_{ij})$ for every pair $k,k' \in (i \leftrightsquigarrow j)$. To explain this freedom to tweak these constraints, recall that the primal constraints $Z \ge 0$ were unnecessary. When we omit them in the primal LP, the corresponding dual constraints that we get are equalities.

    \item Unlike the primal LP which ``usually'' has a unique solution, the dual LP ``usually'' has multiple solutions. As a simple example, consider a topology with $3$ nodes, with frequencies $f = [3,1,2]$. The primal optimum value is $4$ due to the unique optimal BST with depths $D = [0,2,1]$. The dual solution is $R_{12} = R_{23} = 1$, $R_{13} = 2$ with freedom to set $Q$ such that $Q_{123} \in [0,2]$, $Q_{321} \in [0,1]$ and $Q_{123}+Q_{321} \ge 2$. In larger topologies, there may be degrees of freedom on the $R$ variables as well. That is, the sum over all $R_{ij}$'s is fixed (the value of the dual LP), but there will be freedom of how to distribute the sum between the specific variables. Sometimes, the dual LP has a unique solution. Tweaking the current example to $f=[2,1,2]$, the unique solution has the same $R$, with a unique $Q$: $Q_{123}=Q_{321}=1$.

    \item Note that $R_{ij}$ participates in two symmetric-looking frequency inequalities: \\$R_{ij} + \sum_{a : i \in (j \leftrightsquigarrow a)}{Q_{jia}} \le f_j$ and  $R_{ij} + \sum_{a : j \in (i \leftrightsquigarrow a)}{Q_{ija}} \le f_i$. Observe that the $Q$ variables are not the same: In the first inequality they represent a directed relationship from $j$ to every node ``behind'' $i$, while in the second inequality it is a direction from $i$ to every node ``behind'' $j$. None of these $Q$ variables is related to the $Q$ variables in the capping inequalities of $R_{ij}$, where the third node is \emph{between} $i$ and $j$, not behind any of them. Also, not every node appears in some $Q$ related to $R_{ij}$: nodes that branch out of the path between $i$ and $j$ do not define a colinear triplet.

    \item Some of the constraints are redundant. To see it best, consider two leaves $i$ and $j$. They do not have $Q$ variables for nodes behind them, therefore we get the two inequalities: $R_{ij} \le f_j$ and  $R_{ij} \le f_i$. Depending on whether $f_i \ge f_j$ or $f_i \le f_j$, one inequality is redundant. More generally, even when $j$ is a leaf and $i$ is not, we could have redundancy if $f_j \le f_i$. Explicitly, in this case $R_{ij} \le f_i$ is redundant because it is implied by $R_{ij} \le f_j - \sum_{a : i \in (j \leftrightsquigarrow a)}{Q_{jia}} \le f_j \le f_i$.

    \item One can loosely interpret $f_j$ as an amount of tokens to be endowed by $j$ to $i$ and any node behind $i$. Note that there is an inequality for $j$ with every node, so if we think of the topology as tree rooted in $j$ (as explained in item~\ref{dual_rooting}), $j$ has the same amount to endow to each subtree, with overlaps between subtrees. $R_{ij}$ is the amount endowed to $i$, $Q_{jia}$ is the amount endowed to $a$, and full endowment means equality instead of inequality. The fact that $R_{ij} \equiv R_{ji}$ means that $i$ and $j$ must endow the same amount to each other. Because of the multiple budgets of each node and the overlap, this interpretation is not simply some kind of flow.
\end{enumerate}

\subsection{A Weak-Duality Invariant over STT Subtrees}
\label{subsection_dual_LP_analysis_BST}

The \emph{weak duality theorem} states that the dual value is bounded by the primal value. In this subsection we present this relation in an explicit STT-intuitive way. Specifically, given a fixed STT, $T$, we show that for each subtree $T_i$ of $T$, rooted at $i$, the contribution to the dual value of the variables $R_{ab}$ for $a$ and $b$ whose LCA (lowest common ancestor) is $i$, is bounded by the total frequency of the nodes in this subtree except for $i$ (which is the contribution to the primal value, of pushing the nodes one edge deeper when hanging the subtrees of $i$ under it). We denote $T'_i \equiv T_i \setminus\{i\}$:
$$ \sum_{\substack{a,b \in T'_i \\ LCA(a,b)=i}}{R_{ab}} =
\frac{1}{2} \cdot \sum_{a \in T'_i} \sum_{\substack{b \in T'_i \\ LCA(a,b)=i}} {R_{ab}}
\underset{(1)}{=} \frac{1}{2} \cdot \sum_{a \in T'_i} \sum_{\substack{b \in T'_i \\ LCA(a,b)=i}} {(Q_{aib} + Q_{bia})}
$$
$$
= \frac{1}{2} \cdot \sum_{a \in T'_i} \Big ( \sum_{\substack{b \in T'_i \\ LCA(a,b)=i}} {Q_{aib}} \Big ) + \frac{1}{2} \cdot \sum_{b \in T'_i} \Big ( \sum_{\substack{a \in T'_i \\ LCA(a,b)=i}} {Q_{bia}} \Big )
$$
$$
\underset{(2)}{\le} \frac{1}{2} \cdot \sum_{a \in T'_i}{(f_a - R_{ia})} + \frac{1}{2} \cdot \sum_{b \in T'_i}{(f_b - R_{ib})}
= \sum_{a \in T'_i}{f_a} - \sum_{a \in T'_i}{R_{ia}} \ ,
$$
where $(1)$ is by the capping equalities and $(2)$ is by the frequency inequalities. The factor of $\frac{1}{2}$ is due to double-counting in the double-sum. Re-arranging the inequality we get:
$$ \sum_{\substack{a,b \in T_i \\ LCA(a,b)=i}}{R_{ab}} \le \sum_{a \in T'_i}{f_a}
$$
This means that the increase in the dual value, when we ``step up'' from subtrees of $i$ to the subtree rooted at $i$ (that is we add $R_{ab}$'s for pair $a$ and $b$ whose LCA is $i$, including when one of them is $i$) is at most the delta of the primal value which follows from increasing the depth of each node in $T'_i$ by $1$ (sum of frequencies, right-hand side). If we achieve equality in every step then we can conclude that $T$ is in fact an optimal STT.

If $T$ is arbitrary, the answer could be negative. We already know that some topologies are such that the primal optimum is not achieved by any STT, but let us give a simpler example. Consider the case where $f_r = 0$ and $r$ has a single child. Recall our context: given some STT, we focus on the subtree rooted at the node $r$. Then since $\forall i: R_{ri} + \sum{Q_{ria}} \le f_r = 0$, we get that $R_{ri} = 0$ and the dual contribution due to $R$ variables that involve $r$ cannot match the sum of frequencies of the (strict) descendants of $r$.

However, the hope is to be able to do so in the special case where an STT does achieve the optimum, and for this specific STT. In other words, when the structure is not arbitrary but based on what we know of the primal solution. In some cases it is trivial to achieve this delta, one such case is when $f_r$ is heavy enough such that $f_r \ge \sum_{i \in S}{f_i}$ for every $S$ that is an STT subtree rooted at a child of $r$. Assigning $R_{ri} = f_i$ and $Q_{rij} = f_j$ for every $i,j \in S$ yields a feasible solution, and the delta value of the LP is indeed $\sum_{i \in S}{R_{ri}} = \sum_{i \in S}{f_i}$. This assignment is clearly feasible since all the values are non-negative, capping holds because $\forall k \in (r \leftrightsquigarrow i): R_{ri} = f_i = Q_{rki} \le Q_{rki} + Q_{ikr}$, and frequency constraints hold by our assumption, indeed $\forall i: R_{ri} + \sum_{a:i \in (r \leftrightsquigarrow a)}{Q_{ria}} = f_i + \sum_{a:i \in (r \leftrightsquigarrow a)}{f_a} \le \sum_{a \in S_i}{f_a} \le f_r$ where $S_i$ is the subtree rooted at a child of $r$, that contains $i$. This example is no surprise, since a very heavy node should be the root of the STT in the primal solution.

\section{Additional Miscellaneous Results and Discussions}
In this section we list additional insights and results which were found while studying the LP's polytope, some analytical, some via computer analysis.

\subsection{Refining the LP to More Accurately Describe Search Trees}
\label{section_refining_the_LP}
In the discussion leading to Definition~\ref{definition_LP_program}, we noted that the LP does not fully capture the behavior of STTs. In the following, we list a few ``strong'' properties of STTs which were not taken into account, show how to weakly account for them, and finally explain why adding them to the definition of the LP does not remove all the non-integer vertices.

\begin{enumerate}
    \item Path monotonicity: Consider three nodes on a path, in order $a,b,c$. Then $X_{ac} = 1$ only if $a$ is the LCA of itself and $c$, and therefore if $X_{ac} = 1$ then   $X_{ab} = 1$. We get that $X_{ab} \ge X_{ac}$. More generally, if $a,b_1,b_2,\ldots,b_k$ is a path in the topology, $X_{ab_1} \ge X_{ab_2} \ge \ldots \ge X_{ab_k}$. As an example, on the path topology, the ancestry matrix $X$ is monotone-weakly-decreasing in each row, from the main diagonal to the left and to the right. Note that for pure STTs this path monotonicity implies a step function from $0$ to $1$ to $0$ where up to two of these segments may be empty, in particular the row may be constant. In general feasible solutions the steps may be more gradual, but the up-then-down monotonicity is preserved.

    \item Ancestry transitivity: in STTs, if $i$ is an ancestor of $j$ and $j$ is an ancestor of $k$, then $i$ is an ancestor of $k$. This can be written in the non-linear and non-convex form: $X_{ik} \ge \min(X_{ij},X_{jk})$. A linear compromise would be $X_{ik} \ge X_{ij} + X_{jk} - 1$. (While this constraint holds for any triplet, if $k \in (i \leftrightsquigarrow j)$ it is implied by path monotonicity, because then $X_{ik} \ge X_{ij}$, and $X_{jk} \le 1$ for a minimization objective.)

    \item LCA separation: For $k \in (i \leftrightsquigarrow j)$, if $k$ is an ancestor of $i$, then $i$ cannot be an ancestor of $j$. Assume that $i$ is an ancestor of $j$ too, then by the transitivity of ancestry $k$ would be an ancestor of $j$ as well, but since $k$ is between $i$ and $j$, it separates them in contradiction to $i$ being an ancestor of $j$. We can represent this in a linear form as $X_{ki} + X_{ij} \le 1$. Furthermore, we can also require $X_{ji} + X_{ij} \le 1$ for any STT, thus overall: $\forall k \in (i \leftrightsquigarrow j]: X_{ki} + X_{ij} \le 1$.

    \item Refining $Z$: we originally defined $Z_{kij} \le X_{ki},X_{kj}$ as a linearization of $Z_{kij} = \min \{ X_{ki},X_{kj} \} $. Considering STTs for which $X$ has $0/1$ values, we can restrict further and also demand: $Z_{kij} \ge X_{ki} + X_{kj} - 1$.

    \item Row-Min/Column-Max: Consider the ancestry matrix $X$ of some STT with root $r$. The minimum in each row $i \ne r$ is $0$ because $X_{ir} = 0$, and the minimum in row $r$ is $1$ because $\forall j: X_{rj} = 1$. Similarly, the maximum in each column $i \ne r$ is $1$, and $0$ in column $r$. We can define new variables $m_i$ for the minimum in row $i$ and $M_i$ for the maximum in column $i$, and add the constraints: $\forall j: X_{ji} \le M_i$, $\forall j: m_i \le X_{ij}$ and $\sum_{i}{M_i} = n-1$, $\sum_{i}{m_i} = 1$. (The counterparts max-row and min-column are less useful because they change for $i \ne r$ depending on the STT.)
\end{enumerate}

So far it looks very promising, few natural constraints of STTs were missing from the LP, which we can add. Moreover, we introduced only $O(|X|) + |Z|$ extra inequalities, so the problem is still polynomial.\footnote{$|X|$ and $|Z|$ denote the numbers of $X$ and of $Z$ variables, respectively.} However, these extra inequalities do not fix the polytope, and non-integer vertices remain plentiful where they exist. It may remove some, and create others (due to new hyperplanes), but it does not remove all of them which is what we hoped for, if this was an extended formulation. To see this, we only need to go back to our main example, the vertex $P$ in Figure~\ref{figure_explicit_counter_example}. It satisfies all the additional constraints, as we prove:

\begin{enumerate}
    \item Path monotonicity: $X_{12} = X_{13} = \frac{1}{2} > X_{14} = X_{15} = X_{16} = X_{17} = 0$; $X_{23} = 1 > X_{24} = X_{25} = X_{26} = \frac{1}{2} > X_{27} = 0$; $X_{3i} = 0$ for all $i \ne 3$; $X_{43} = 1 > X_{42} = X_{41} = X_{46} = \frac{1}{2} > X_{47} = 0$; $X_{54} = X_{53} = \frac{1}{2} > X_{52} = X_{51} = X_{56} = X_{57} = 0$; $X_{63} = 1 > X_{6i} = \frac{1}{2}$ for all $i \ne \{3,6\}$; $X_{7i} = \frac{1}{2}$ for all $i \ne 7$.
    
    \item Ancestry transitivity: note that if $X_{ij},X_{jk} \le \frac{1}{2}$ then the inequality becomes trivial because $X_{ik} \ge 0$ anyway. So the only entries that can possibly make $P$ infeasible are pairs that include at least one of the entries $X_{23}=X_{43}=X_{63}=1$. If $i \in \{2,4,6\}$ and $j=3$ then $X_{jk}=0$ regardless of $k$ and we get the constraint $X_{ik} \ge 0$ (always true). If $j \in \{2,4,6\}$ and $k=3$ we need to check that $X_{i3} \ge X_{ij}$ for any $i \notin \{j,3\}$. This holds because $X_{i3} \ge \frac{1}{2} \ge X_{ij}$ (here $j \ne 3$).
        
    \item LCA separation: If $X_{ki},X_{ij} \le \frac{1}{2}$ then $X_{ki} + X_{ij} \le 1$ is satisfied. Then, to verify that $P$ is feasible, we only need to consider the case that $i=3$ or $j=3$. If $i=3$ then $X_{ij} = 0$ for any $j$ and the inequality is satisfied. For $j=3$ we consider $i \in \{2,4,6\}$, which is a neighbor of $j$, hence $k = j = 3$ is the only option for $k$, and the inequality holds.

    \item Refining $Z$: If $X_{ki},X_{kj} \le \frac{1}{2}$ then $Z_{kij} \ge 0$ is trivial. Then focus on $i=3$ ($j=3$ is symmetric), with $k \in \{2,4,6\}$ and the inequality reduces to $Z_{k3j} \ge X_{kj}$. The requirement $k \in (i \leftrightsquigarrow j)$ together with $i=3$ implies that $(k,j) \in \{ (2,1), (4,5), (6,7) \}$. Then we can verify that $Z_{213}  = X_{21} = Z_{435} = X_{45} =Z_{637}= X_{67}   = \frac{1}{2}$ hence the inequalities hold.

    \item Row-Min/Column-Max: the min-row vector is $\vec{q} = [0,0,0,0,0,\frac{1}{2},\frac{1}{2}]$, hence we can choose $\vec{m} = \vec{q}$ to satisfy the row-min part. The max-column vector is $\vec{Q} = [\frac{1}{2},\frac{1}{2},1,\frac{1}{2},\frac{1}{2},\frac{1}{2},\frac{1}{2}]$, and there are plenty of choices for $M$ such that $\forall i: M_i \ge Q_i$ and $\sum_i{M_i} = n-1$, which satisfy the column-max part.
\end{enumerate}

In conclusion, even though we refined the LP to capture more accurately the behavior of STTs, we did not get an extended formulation that removes every non-STT vertex, and an integrality gap remains. One reason may be that we did not capture all possible refinements. Indeed, when we compute the convex hull of the $X$s induced by STTs, there are many additional constraints. To name just one such family, for every permutation $\pi$ without fixed points ($\forall i: \pi_i \ne i$), $\sum_i{X_{i,\pi_i}} \ge 1$. The intuitive interpretation is that some node must be the root $r$ of the STT, therefore $X_{r,\pi_r} = 1$. The size of this set is roughly $\frac{n!}{e}$.

\subsection{Eliminating the \emph{Z} Variables from the Formulation}
\label{subsection_with_and_without_z_fourier_motzkin}
In the LP definition in Section~\ref{subsection_linear_program}, we used $Z_{kij}$ variables whose conceptual purpose is to represent an LCA relationship of $k$ over $i$ and $j$, practically a linear replacement to a term $\min\{X_{ki},X_{kj}\}$ which cannot simply be used in an LP formulation. It is natural to ask whether we can eliminate these variables and replace them by additional constraints on the $X_{ij}$ variables. This may have drastic effects, in particular making many integer vertices of the original LP polytope non-vertices, as discussed in Remark~\ref{remark_int_vertices_disappear_when_no_Z}. To get rid of all the $Z$ variables, we can take every constraint that  contains $d$ minimum-terms (or rather $d$ $Z$-variables each representing the minimum of a pair of $X$-variables), and replace it by $2^d$ inequalities where each inequality corresponds to a choice of a particular $X$-variable from each  minimum-pairs as in Example~\ref{example_Z_elimination}.\footnote{We can formally and compactly define all the inequalities as follows: Let $G$ be the set of $2^d$ functions from $(i \leftrightsquigarrow j)$ to $\{i,j\}$. Then for each $g \in G$ we get the inequality $X_{ij}+X_{ji} + \sum_{k \in (i \leftrightsquigarrow j)}{X_{k,g(k)}} \ge 1$.}
Even though we get an exponential number of constraints, we can still use a separation oracle (cutting-plane oracle) to solve the LP in polynomial time, since identifying the tightest inequality out of each such group of $2^d$ inequalities that replace a single original inequality is easy:
 just take the inequality where the smallest of each pair $X_{ki}$ and $X_{kj}$ participates.

\begin{example}
\label{example_Z_elimination}
Consider an ancestry constraint on a pair $a$ and $d$ with two nodes $b,c$ on the path between them: $X_{ad} + X_{da} + \min\{ X_{ba},X_{bd} \} + \min\{ X_{ca},X_{cd} \} \ge 1$. The constraint is satisfied if and only if the following set of four constraints is satisfied:
$X_{ad} + X_{da} + X_{ba} + X_{ca} \ge 1$ and
$X_{ad} + X_{da} + X_{ba} + X_{cd} \ge 1$ and
$X_{ad} + X_{da} + X_{bd} + X_{ca} \ge 1$ and
$X_{ad} + X_{da} + X_{bd} + X_{cd} \ge 1$.
\end{example}

The Fourier-Motzkin elimination technique~\cite{BookForFourierMotzkin} can also be used to eliminate the $Z$ variables. One can verify that its application yields the same set of inequalities.

Observe that once we eliminate the $Z$ variables, ancestry constraints must be inequalities. Previously, we could make them equalities $(=1)$, but now we cannot do this since we do not know which among the $2^d$ constraints is tight. Overall, the trade-off of eliminating $Z$ gives us a more complex presentation of the polytope (more inequalities) in exchange for less dimensions (no $Z$) as well as less vertices overall as it turns out.

\begin{remark}
\label{remark_int_vertices_disappear_when_no_Z}
It is interesting to mention that the proof of Theorem~\ref{theorem_non_tree_vertex} breaks under the new formulation of the LP without $Z$ variables. There, the second case in the proof completely disappears because there are no $Z$ variables to manipulate. The first case of cyclic ancestry is feasible, but it might not be a vertex anymore. For an example of vertices that become non-vertices without $Z$, revisit Table~\ref{table_vertices_n3}: The two ``LCA abuse'' vertices become dominated by the vertex that corresponds to the STT whose root is $2$ when we eliminate $Z$. More generally, the set $\mathbb{P}'$ as defined in the proof of Theorem~\ref{theorem_non_tree_vertex} gets much smaller or even empty, and no longer contains all the STT induced vertices.
\end{remark}

Overall, it is probably better to eliminate the $Z$ variables from the LP. It is ``more true'' (by capturing the minimum as intended), and it also affects the density of non-integer vertices, as we explain next. As a warm-up, note that given that some nodes disappear, the relative density of integer versus non-integer nodes may change. We can refine this differentiation by categorizing each vertex according to its highest denominator. In this case, consider for example the path topology over $n=5$ nodes. By enumerating the vertices of its corresponding polytope, we find that without eliminating $Z$, there are $5983$ integer vertices (out of which only $42$ are due to STTs), $3886$ half-integer vertices, and $76$ third-integer vertices. Compactly, we write $[5983,3886,76]$.\footnote{There are no vertices with both half-integer and third-integer coordinates. This scenario is possible in vertices for larger topologies, by Theorem~\ref{theorem_combine_trees}.
% see Example~\ref{example_mixed_halves_and_thirds}.
Also, none of these non-integer vertices is projected onto a vertex in $D$-space, so in $D$-space they are all non-vertices.} Counting the vertices of the polytope in the $Z$-eliminated version, we get much fewer vertices, distributed $[519,158,7]$ (still $42$ STT vertices, out of the $519$). Observe that, for example, the percentage of non-integer vertices drops from $40\%$ to $24\%$, and that third-integer vertices rise from $0.76\%$ to $1.02\%$. Changed percentages may affect our probability of finding vertices
when we use a random direction as the objective.
%\footnote{A uniformly random direction may not result in a uniformly random vertex, but the discussion is mostly informal at this point anyway.}

\begin{table}[!ht]%[!ht]
    \scriptsize
    \begin{center}
    \begin{tabular}{|c|l|l||l|l||c|c|}

    \hline
    \multicolumn{1}{|c|}{} & \multicolumn{2}{|c||}{$(X,Z,D)$-Direction} & \multicolumn{2}{|c||}{$(X,D)$-Direction} & \multicolumn{2}{|c|}{$D$-Direction} \\
    \hline
    Path length ($n$) &
    LP &
    $Z$-elim. &
    LP &
    $Z$-elim. &
    LP &
    $Z$-elim. \\
    \hline
    \hline
    6 & [1, 2] & [1, 2] & [1, 2] & [1, 2] & [1] & [1] \\
    \hline
    7 & [1, 2, 3, 4] & [1, 2] & [1, 2] & [1, 2] & [1] & [1] \\
    \hline
    8 & [1, 2, 3, 4, 5] & [1, 2, 3] & [1, 2, 3] & [1, 2, 3] & [1] & [1] \\
    \hline
    9 & [1, 2, 3, 4, 5] & [1, 2, 3] & [1, 2, 3] & [1, 2, 3] & [1] & [1] \\
    \hline
    10 & [1, 2, 3, 4, 5, 6] & [1, 2] & [1, 2] & [1, 2] & [1] & [1] \\
    \hline
    11 & [1, 2, 3, 4, 7] & [1, 2, 3] & [1, 2, 3] & [1, 2, 3] & [1] & [1] \\
    \hline
    12 & [1, 2, 3, 4, 5, 6, 7, 8, 9, 11, 12, 14] & [1, 2] & [1, 2] & [1, 2] & [1] & [1] \\
    \hline
    13 & [1, 2, 3, 4, 5, 6, 7, 8, 9, 11, 12, 16, 20] & [1, 2] & [1, 2] & [1, 2] & [1] & [1] \\
    \hline
    14 &  [1, 2, 3, 4, 5, 6, 8, 9, 10,11, 12, 13,18, 20, 23, 43] & [1, 2] & [1, 2] & [1, 2] & [1] & [1] \\
    \hline
    15 & \begin{tabular}{@{}l@{}}[1, 2, 3, 4, 5, 6, 7, 8, 9, 10, 11, 12, 14, 16, 20, \\ 22, 23, 28, 40, 47, 80] \end{tabular} & [1, 2, 3] & [1, 2, 3] & [1, 2, 3] & [1] & [1] \\
    \hline
    16 & 
    \begin{tabular}{@{}l@{}l@{}}[1, 2, 3, 4, 5, 6, 7, 8, 9, 10, 11, 12, 13, 14, 15, \\ 17, 18, 20, 22, 23, 24, 26, 28, 34, 36, 39, 46, 58, \\ 67, 87, 92, 111, 134, 174, 257, 261, 522, 999] \end{tabular}
     & [1, 2, 3] & [1, 2, 3] & [1, 2, 3] & [1] & [1] \\
    \hline
    \end{tabular}
    \end{center}
    \caption{\footnotesize{A summary of denominators of vertex coordinates found when solving the LP for path topologies, of different lengths, in six variations. The vertices were found by sampling directions, which is not exhaustive as can be seen since for $n=5$ we know of third-integer vertices in $(X,Z,D)$-space, but here for $n=6$ we have not found such vertices. We chose the objective (weights) in either: (1) a general direction with arbitrary weights for $(X,Z,D)$, (2) $(X,D)$-direction with weights $0$ for all $Z$, and (3) depths-only direction with non-zero weights only for $D$. In each of these three methods, we solve the LP as in Definition~\ref{definition_LP_program} (LP column), and its variant where we eliminate $Z$ as in Section~\ref{subsection_with_and_without_z_fourier_motzkin} ($Z$-elim. column). It is clear that an arbitrary direction reveals fractions with increasing denominators, yet they seem to be artifacts of the presence of the $Z$, and are either non-existent or very rare when eliminating $Z$ or when considering $(X,D)$ or $D$ directions. Furthermore, all the vertices found in $D$-directions are integer. Among the last five columns, $2$ and $3$ denominators never appear together, i.e., every vertex is either integer, half-integer or third-integer. However, this fact is an artifact of sampling the rare mixed cases, which exist by Theorem~\ref{theorem_combine_trees} (e.g., combine two paths of length $n=5$ with an extra node to get a path of $n=11$).}}
    \label{table_path_tree_fractions}
\end{table}

We cannot enumerate all the vertices in larger topologies, and the interesting question is whether the ``\emph{fractionality}'' of the polytope is affected only by changed percentages, or whether some types of fractions completely disappear while others are introduced when we eliminate $Z$. By solving the LP in random directions, it looks like there is indeed a strong difference in behavior. We focused our numerical analysis on topologies that are paths, see Table~\ref{table_path_tree_fractions}. There, when solving the LP with objectives in arbitrary directions ($(X,Z,D)$ may all have positive weights) we find that the denominators grow when the topology grows. However, when the objective is only in $(X,D)$ directions, 
which are the only directions that remain when
we eliminate $Z$ from the LP, the fractionality only contains integer, half-integer and third-integer vertices. Finally, when solving either version of the LP in a pure $D$-direction, which is the ``intended'' kind of objective, we only find integer vertices. These results hint that something is inherently ``less fractional'' when we eliminate $Z$, even if they are based on sampling rather than a full enumeration or an analytic proof. These results also hint at the possibility that the LP perhaps finds an optimal STT when path topologies are considered.

\subsection{Are Path-monotonicity Constraints Helpful?}
In Section~\ref{section_refining_the_LP} we listed several possible refinements of the LP. Among them, the set of path-monotonicity constraints is arguably one of the most natural and simple to grasp. We have already seen that adding these constraints does not remove all non-integer vertices, but one can wonder whether it helps. It turns out to be a double-edged sword.

At first glance, it has benefits. The small topologies with non-integer vertices that we found have $74$ such vertices in total after sieving symmetric copies (revisit Table~\ref{table_trees_summary_analysis_vertices_etc_FEW}). For each specific topology, more than half satisfy these constraints, so no topology becomes ``truly nicer''. Testing singular vertices is sort of cheating, because we need to see how the polytope changes as a whole. When considering the path over $5$ nodes, we get promising results. Consider the version without $Z$ variables, as discussed in detail in Section~\ref{subsection_with_and_without_z_fourier_motzkin}. When we add path-monotonicity constraints, we still have non-integer vertices overall, but instead of having $165$ as such ($158$ half-integer and $7$ third-integer), the new polytope only has $4$ non-integer vertices, all half-integer. This is not quite fully-integer, but it looks like a promising simplification.

Alas, some topologies are badly affected. Consider topology $U_{(5,1)}$ ($5$ nodes, ``T''-shaped) which was found to have an integer polytope for the original LP. If we add path-monotonicity constraints, its $(X,Z,D)$-polytope now has half-integer and third-integer vertices. So, to conclude this discussion: while path-monotonicity seems like a clean and simple addition to the LP, it is not clear whether it helps or complicates things.

\subsection{Multiple Vertices with the Same Depths}
Throughout the paper, we sometimes discuss vertices in $(X,Z,D)$-space, and sometimes in $D$-space. In particular, the normals method uses $D$-space to solve and find ``$(X,Z,D)$-vertices'', which we then project back to $D$-space. It is natural to ask whether the projection to $D$-space results in collisions of vertices, that is, if there exist $(X,Z,D)$-vertices, $P \ne P'$, such that $P_D = P'_D$. Because the LP without $Z$ variables is arguably nicer and ``more true'' (recall Section~\ref{subsection_with_and_without_z_fourier_motzkin}), let us wonder about collisions of vertices of the $(X,D)$-polytope of the LP when projected to $D$-space.

For STTs the answer is negative: given a depths-vector $D$ induced by an STT, there is exactly one $(X,D)$-vertex in its pre-image. Indeed, the depths-vector $D$ uniquely implies the value of all the $X$ variables, because there is a unique index $r$ (the root) such that $D_r=0$, so it must be that $\forall i \ne r: X_{ir}=0,X_{ri}=1$. Then we recurse in each connected component, now looking for a vertex of depth one in each connected component etc.

In contrast, general vertices may collide. When we solved the LP without the $Z$ variables, in all the primary directions that discover non-STT vertices, we found a few collisions. Recall that the objectives that we used had non-zero coefficients only for the $D$-variables, so to truly study the collisions one should use general objectives. Furthermore, note that refinements of the LP such as those discussed in Section~\ref{section_refining_the_LP} may affect these collisions.
% [Slightly more information:] The collisions that we find are only of pairs, but in general there might be collisions of higher cardinality. The collisions that we find (depends on the solver being used since the objective function wsa in a pure $D$-direction): Two collisions for the topology $U_{(8,5)}$, one for $U_{(8,6)}$, and eight for $U_{(8,13)}$. Note that these findings only prove that collisions happen and say nothing about other collisions that we missed or the cardinality of these collisions (may be more than just pairs).

\section{Open Questions}
\label{section_open_questions}

In this section we summarize the main open questions regarding the LP approach towards solving Problem~\ref{problem_main_problem_stt}. Of course, solving the problem with any tool/approach, is still open. Moreover, we assumed that all queries are successful, that is, there are no searches of non-existing nodes. If the optimum can be found for this case, it may also be interesting to consider an extension with search failures in a way analogous to BSTs, where the misses occur over an edge (between neighbors) or ``beyond'' a leaf, with their own frequencies.
%%% === Welcome to the source code. Some of Yaniv's extra thoughts and discussions of the following failed attempts are detailed (as comments) in the file 'section_open_questions_extra_thoughts.tex'. ===

\paragraph{Optimality Questions and Alternative Approaches}
\begin{enumerate}
    \item Can we extend the LP, or formulate an alternative which could be shown to have no integrality gap? Can the dual formulation (Section~\ref{section_dual_LP}) be useful?
    
    \item Can more general tools help? For example Quadratic Programming (QP) or Semi-Definite Programming (SDP)?

    \item Is it useful to view the optimization problem in bilinear terms of minimizing the expression $\mathbf{1}^T \cdot X \cdot f$ where $\mathbf{1}^T$ is an all-ones row vector, $f$ is the frequencies vector, and $X$ is the matrix of the $X_{ij}$ variables subjected to the LP constraints?

    \item The paper \cite{DPtoLP2011} gives a general way to transform dynamic programming (DP) formulations to linear programs (LP), thus one may take Knuth's DP  for BSTs~\cite{KnuthOptStatic1971}, or the $k$-cut trees DP of \cite{SplayTreesonTrees_andPTAS}, and try to generalize it to all STTs. For the details of the latter DP, see Chapter 5 in the thesis of Berendsohn~\cite{BerendsohnThesis}.

    \item \label{item_subclass} Can we prove that the LP approach with root rounding is optimal for a sub-class of topologies other than stars (Theorem~\ref{theorem_star_has_integer_vertices})? The most interesting and promising candidates are probably (a) the class of paths, and (b) the class of graphs with edge-diameter of $3$ (``almost stars'') for which we do not know if the LP itself is integer.\footnote{
    For edge-diameter $4$ and above we know that: The full LP is known to be non-integer on the path topology of $5$ nodes. This is discussed in more detail in Section~\ref{subsection_with_and_without_z_fourier_motzkin}. Moreover, the $D$-space projection of the LP on the $n=7$ nodes topology in Theorem~\ref{theorem_non_integer_optimum_long_star} is also non-integer. Both topologies have edge-diameter $4$, and can be extended to topologies with larger edge-diameter that have non-integer solutions.}
    
    \item Are there necessary or sufficient conditions on a topology such that the projection of its LP to $D$-space is integral (and thus optimal)? Figure~\ref{figure_all_trees_up_to_n8} shows that many non-integer LPs do become integer when projected.
\end{enumerate}

\paragraph{Improved Approximations and Rounding Alternatives}
\begin{enumerate}
    \item What is the true integrality gap of the LP? Denote it by $\rho$. We know that $\rho \le 2$ by Property~\ref{property_tree_2_approx} and that $\rho \ge \frac{95}{93} \approx 1.02$ by Table~\ref{table_integrality_gaps}. Can we tighten these bounds?

    \item Can the root rounding (Definition~\ref{definition_rounding_scheme}) be refined, to define tie breaking more carefully, to guarantee an integrality gap better than $2$?

    \item Is there a way to take the frequencies into account when rounding? Observe that the root rounding is agnostic, which is somewhat less natural. Studying the dual LP in this context may be relevant since its polytope depends on the frequencies.

    \item \label{item_peeling} Is there an alternative rounding scheme which is better (absolutely, or by being easier to analyze)? As a concrete suggestion, is there a way to take a feasible solution and (almost) decompose it to STT induced points? For example, ``peel'' off STTs one by one until the remainder is small. Then we can choose the best STT in the decomposition. Another alternative decomposition rounding would be to maintain a collection of STTs whose convex combination is ``close'' to the solution we wish to round. Start by choosing a single node of the topology (a singleton STT is exact for it), then extend the topology by one node at a time, where in each step we update the set of STTs by adding the new node to each STT, possibly forking some of them if we need different STTs that are the same up to the location of the new node. As an example, the non-integer vertex $P$ shown in Figure~\ref{figure_explicit_counter_example} with depths $D^P = (2,2,4.5,2,2,1.5,0.5)$ can be presented as almost-average of STTs $B$ and $C$ with depths vectors $D^B = (2,3,4,1,2,0,1)$, $D^C = (2,1,5,3,2,4,0)$ ($P$ is better: $D^P = \frac{1}{2}(D^B+D^C) - (0,0,0,0,0,\frac{1}{2},0)$).

    \item Is there an iterative rounding scheme that improves the approximation or even achieves optimality? The root rounding (Definition~\ref{definition_rounding_scheme}) takes a fractional solution of the LP, and rounds it fully. Iterative rounding does not help the root rounding, but perhaps it could help in other ways. For example, if we could determine a good choice of root, we could apply this method repeatedly, $O(n)$ times, each time fixing the root of a subtree of the resulting STT. This will still take polynomial time in total. Ideally we could hope to determine an optimal choice for the root, but any choice that would lead to an approximation ratio better than $2$ would improve the state of the art.
    %One major obstacle with iterative rounding is that we need to be able to re-define a smaller problem. Unless we determine a good root and then break the topology to sub-components, it is not clear how to reduce the problem.
    A related question is whether there is a way to remove an arbitrary (non-root) node, or contract an edge of the topology and recur on the smaller problem?
    
\end{enumerate}

 \paragraph{Miscellaneous LP Properties}
\begin{enumerate}
    \item Considering the LP version without the $Z$ variables (Section~\ref{subsection_with_and_without_z_fourier_motzkin}), every coordinate of every vertex that we found was either integer, half-integer, or third-integer. When considering vertices that are solutions of the LP in a $D$-direction, we only encounter integer and half-integer coordinates (Remark~\ref{remark_new_vertices_are_non_integer_half_integer_only}). Can we prove that these properties hold in general?

    \item Recall Definition~\ref{definition_partially_integer_vertex} of partially integer vertices: these are $(X,Z,D)$ vertices whose $D$ coordinates are integer. In Example~\ref{figure_Example_partially_interger_vertex} we present such a vertex, yet when projected to $D$-space it is no longer a vertex. Are there partially integer vertices that remain vertices under the projection? Property~\ref{property_integer_domination} does not rule them out. It would be interesting to find out either way.

\end{enumerate}

\section*{Acknowledgments}
We thank Anupam Gupta for useful discussions.
\bibliography{reference}

\appendix
\section{Appendix: Extras}
\label{section_appendix_extras}

\subsection{Table for all Small Topologies}

Table~\ref{table_trees_summary_analysis_vertices_etc_FULL} extends Table~\ref{table_trees_summary_analysis_vertices_etc_FEW} for all tree topologies of size $n \le 8$. For the sake of completeness, Table~\ref{table_trees_edges} explicitly lists the edges of each topology to reduce any ambiguity due to automorphism on the way we present the topologies unnamed in Figure~\ref{figure_all_trees_up_to_n8}.

\begin{table}[!htp]%[!ht]
    \scriptsize
    \begin{center}
    \begin{tabular}{|c|c|c|c|c|c|l|l|}

    \hline
    \begin{tabular}{@{}c@{}}Topology \\ $U_{(n,i)}$ \end{tabular} &
    STTs &
    \begin{tabular}{@{}c@{}}Primary \\ Directions \end{tabular} &
    \begin{tabular}{@{}c@{}}False \\ Facets\end{tabular} &
    \begin{tabular}{@{}c@{}}Frac \\ Vs\end{tabular} &
    \begin{tabular}{@{}c@{}}Frac Vs \\ Classes\end{tabular} &
    \begin{tabular}{@{}c@{}}$D$-space \\ denom. \end{tabular} &
    \begin{tabular}{@{}c@{}}$XZD$-space \\ denom. \end{tabular} \\
    \hline
    \hline
    (3,0) & \ \ \ \ 5 & \ \ \ \ 9 & \ \ 0 & . & . &\{1\} & \{1\} {*} \\
    \hline
    (4,0) & \ \ \ 14 & \ \ \ 32 & \ \ 0 & . & . &\{1\} & \{1\} {*} \\
    \hline
    (4,1) & \ \ \ 16 & \ \ \ 32 & \ \ 0 & . & . &\{1\} & \{1\} {*} \\
    \hline
    (5,0) & \ \ \ 42 & \ \ 145 & \ \ 0 & . & . &\{1\} & \{1, 2, 3\} {*} \\
    \hline
    (5,1) & \ \ \ 51 & \ \ 152 & \ \ 0 & . & . &\{1\} & \{1\} {*} \\
    \hline
    (5,2) & \ \ \ 65 & \ \ 161 & \ \ 0 & . & . &\{1\} & \{1\} {*} \\
    \hline
    (6,0) & \ \ 132 & \ \ 776 & \ \ 0 & . & . &\{1\} & \{1, 2\} \\
    \hline
    (6,1) & \ \ 166 & \ \ 910 & \ \ 0 & . & . &\{1\} & \{1, 2\} \\
    \hline
    (6,2) & \ \ 176 & \ \ 908 & \ \ 0 & . & . &\{1\} & \{1, 2\} \\
    \hline
    (6,3) & \ \ 214 & \ \ 949 & \ \ 0 & . & . &\{1\} & \{1\} \\
    \hline
    (6,4) & \ \ 236 & \ \ 978 & \ \ 0 & . & . &\{1\} & \{1\} \\
    \hline
    (6,5) & \ \ 326 & \ 1071 & \ \ 0 & . & . &\{1\} & \{1\} \\
    \hline
    (7,0) & \ \ 429 & \ 4839 & \ \ 0 & . & . &\{1\} & \{1, 2, 3, 4\} \\
    \hline
    (7,1) & \ \ 552 & \ 5932 & \ \ 0 & . & . &\{1\} & \{1, 2, 3\} \\
    \hline
    (7,2) & \ \ 605 & \ 6224 & \ \ 0 & . & . &\{1\} & \{1, 2, 3\} \\
    \hline
    (7,3) & \ \ 662 & \ 6364 & \ 39 & 9 & 2 & \{1, 2\} & \{1, 2\} \\
    \hline
    (7,4) & \ \ 836 & \ 6817 & \ \ 0 & . & . &\{1\} & \{1, 2\} \\
    \hline
    (7,5) & \ \ 807 & \ 7002 & \ \ 0 & . & . &\{1\} & \{1, 2\} \\
    \hline
    (7,6) & \ \ 930 & \ 6933 & \ \ 0 & . & . &\{1\} & \{1, 2\} \\
    \hline
    (7,7) & \ \ 721 & \ 7077 & \ \ 0 & . & . &\{1\} & \{1, 2\} \\
    \hline
    (7,8) & \ 1135 & \ 7534 & \ \ 0 & . & . &\{1\} & \{1\} \\
    \hline
    (7,9) & \ 1337 & \ 7579 & \ \ 0 & . & . &\{1\} & \{1\} \\
    \hline
    (7,10) & \ 1957 & \ 8733 & \ \ 0 & . & . &\{1\} & \{1\} \\
    \hline
    (8,0) & \ 1430 & 35097 & \ \ 0 & . & . &\{1\} & \{1, 2, 3, 4, 5\} \\
    \hline
    (8,1) & \ 1870 & 44103 & \ \ 0 & . & . &\{1\} & \{1, 2, 3\} \\
    \hline
    (8,2) & \ 2094 & 46368 & \ \ 0 & . & . &\{1\} & \{1, 2, 3\} \\
    \hline
    (8,3) & \ 2164 & 47535 & \ \ 0 & . & . &\{1\} & \{1, 2, 3\} \\
    \hline
    (8,4) & \ 2416 & 48291 & 362 & 65 & 38 & \{1, 2\} & \{1, 2, 3\} \\
    \hline
    (8,5) & \ 2952 & 56376 & 120 & 2 & 1 & \{1, 2\} & \{1, 2, 3\} \\
    \hline
    (8,6) & \ 2802 & 56724 & \ 10 & 2 & 1 & \{1, 2\} & \{1, 2, 3, 4\} \\
    \hline
    (8,7) & \ 3232 & 57252 & \ \ 0 & . & . &\{1\} & \{1, 2, 3\} \\
    \hline
    (8,8) & \ 2952 & 51172 & \ \ 0 & . & . &\{1\} & \{1, 2, 3\} \\
    \hline
    (8,9) & \ 3490 & 53029 & \ \ 0 & . & . &\{1\} & \{1, 2, 3\} \\
    \hline
    (8,10) & \ 2470 & 53923 & \ \ 0 & . & . &\{1\} & \{1, 2, 3, 4\} \\
    \hline
    (8,11) & \ 3988 & 54201 & \ 78 & 18 & 4 & \{1, 2\} & \{1, 2\} \\
    \hline
    (8,12) & \ 3332 & 56404 & 528 & 60 & 24 & \{1, 2\} & \{1, 2, 3\} \\
    \hline
    (8,13) & \ 4076 & 65733 & 946 & 28 & 4 & \{1, 2\} & \{1, 2\} \\
    \hline
    (8,14) & \ 4674 & 64110 & \ \ 0 & . & . &\{1\} & \{1, 2\} \\
    \hline
    (8,15) & \ 4884 & 62553 & \ \ 0 & . & . &\{1\} & \{1, 2\} \\
    \hline
    (8,16) & \ 3996 & 63179 & \ \ 0 & . & . &\{1\} & \{1, 2, 3\} \\
    \hline
    (8,17) & \ 5940 & 59967 & \ \ 0 & . & . &\{1\} & \{1\} \\
    \hline
    (8,18) & \ 5142 & 58200 & \ \ 0 & . & . &\{1\} & \{1, 2, 3\} \\
    \hline
    (8,19) & \ 6842 & 71285 & \ \ 0 & . & . &\{1\} & \{1\} \\
    \hline
    (8,20) & \ 7284 & 68654 & \ \ 0 & . & . &\{1\} & \{1\} \\
    \hline
    (8,21) & \ 8970 & 68714 & \ \ 0 & . & . &\{1\} & \{1\} \\
    \hline
    (8,22) & 13700 & 83434 & \ \ 0 & . & . &\{1\} & \{1\} \\
    \hline
    \end{tabular}
    \end{center}
    \caption{
    \footnotesize{\largeTablecaption{Summary of all tree topologies up to $n \le 8$ nodes.}}}
    \label{table_trees_summary_analysis_vertices_etc_FULL}
\end{table}

%% We could remove the following table but it is better to have it so that there is no confusion about the names of nodes when we detail specific vectors where the order of coordinate correspond to some specific order of the nodes.
\begin{table}[!ht]%[!ht]
    \scriptsize
    \begin{center}
    \begin{tabular}{|c|c|l|||c|c|l|}

    \hline
    $U_{(n,i)}$ & $D$ & Edges & $U_{(n,i)}$ & $D$ & Edges \\
    \hline
    \hline
    (3,0) & 2 & (1,2), (2,3) & (8,0) & 7 & (1,2), (2,3), (3,4), (4,5), (5,6), (6,7), (7,8) \\
    \hline
    (4,0) & 3 & (1,2), (2,3), (3,4) & (8,1) & 6 & (1,2), (2,3), (2,8), (3,4), (4,5), (5,6), (6,7) \\
    \hline
    (4,1) & 2 & (1,2), (2,3), (2,4) & (8,2) & 6 & (1,2), (2,3), (3,4), (3,8), (4,5), (5,6), (6,7) \\
    \hline
    (5,0) & 4 & (1,2), (2,3), (3,4), (4,5) & (8,3) & 6 & (1,2), (2,3), (3,4), (4,5), (4,8), (5,6), (6,7) \\
    \hline
    (5,1) & 3 & (1,2), (2,3), (2,5), (3,4) & (8,4) & 5 & (1,2), (2,3), (3,4), (3,7), (4,5), (5,6), (7,8) \\
    \hline
    (5,2) & 2 & (1,2), (2,3), (2,4), (2,5) & (8,5) & 5 & (1,2), (2,3), (2,7), (3,4), (3,8), (4,5), (5,6) \\
    \hline
    (6,0) & 5 & (1,2), (2,3), (3,4), (4,5), (5,6) & (8,6) & 5 & (1,2), (2,3), (2,7), (3,4), (4,5), (4,8), (5,6) \\
    \hline
    (6,1) & 4 & (1,2), (2,3), (2,6), (3,4), (4,5) & (8,7) & 5 & (1,2), (2,3), (3,4), (3,7), (4,5), (4,8), (5,6) \\
    \hline
    (6,2) & 4 & (1,2), (2,3), (3,4), (3,6), (4,5) & (8,8) & 5 & (1,2), (2,3), (2,7), (2,8), (3,4), (4,5), (5,6) \\
    \hline
    (6,3) & 3 & (1,2), (2,3), (2,5), (3,4), (3,6) & (8,9) & 5 & (1,2), (2,3), (3,4), (3,7), (3,8), (4,5), (5,6) \\
    \hline
    (6,4) & 3 & (1,2), (2,3), (2,5), (2,6), (3,4) & (8,10) & 5 & (1,2), (2,3), (2,7), (3,4), (4,5), (5,6), (5,8) \\
    \hline
    (6,5) & 2 & (1,2), (2,3), (2,4), (2,5), (2,6) & (8,11) & 4 & (1,2), (2,3), (3,4), (3,6), (3,7), (4,5), (7,8) \\
    \hline
    (7,0) & 6 & (1,2), (2,3), (3,4), (4,5), (5,6), (6,7) & (8,12) & 4 & (1,2), (2,3), (2,6), (3,4), (3,7), (4,5), (7,8) \\
    \hline
    (7,1) & 5 & (1,2), (2,3), (2,7), (3,4), (4,5), (5,6) & (8,13) & 4 & (1,2), (2,3), (2,6), (3,4), (3,7), (4,5), (4,8) \\
    \hline
    (7,2) & 5 & (1,2), (2,3), (3,4), (3,7), (4,5), (5,6) & (8,14) & 4 & (1,2), (2,3), (2,6), (2,7), (3,4), (3,8), (4,5) \\
    \hline
    (7,3) & 4 & (1,2), (2,3), (3,4), (3,6), (4,5), (6,7) & (8,15) & 4 & (1,2), (2,3), (2,6), (3,4), (3,7), (3,8), (4,5) \\
    \hline
    (7,4) & 4 & (1,2), (2,3), (3,4), (4,5), (4,6), (4,7) & (8,16) & 4 & (1,2), (2,3), (2,6), (3,4), (4,5), (4,7), (4,8) \\
    \hline
    (7,5) & 4 & (1,2), (2,3), (3,4), (3,6), (4,5), (4,7) & (8,17) & 4 & (1,2), (2,3), (3,4), (3,6), (3,7), (3,8), (4,5) \\
    \hline
    (7,6) & 4 & (1,2), (2,3), (3,4), (3,6), (3,7), (4,5) & (8,18) & 4 & (1,2), (2,3), (2,6), (2,7), (2,8), (3,4), (4,5) \\
    \hline
    (7,7) & 4 & (1,2), (2,3), (2,6), (3,4), (4,5), (4,7) & (8,19) & 3 & (1,2), (2,3), (2,5), (2,6), (3,4), (3,7), (3,8) \\
    \hline
    (7,8) & 3 & (1,2), (2,3), (2,5), (3,4), (3,6), (3,7) & (8,20) & 3 & (1,2), (2,3), (2,5), (2,6), (2,7), (3,4), (3,8) \\
    \hline
    (7,9) & 3 & (1,2), (2,3), (3,4), (3,5), (3,6), (3,7) & (8,21) & 3 & (1,2), (2,3), (2,5), (2,6), (2,7), (2,8), (3,4) \\
    \hline
    (7,10) & 2 & (1,2), (2,3), (2,4), (2,5), (2,6), (2,7) & (8,22) & 2 & (1,2), (2,3), (2,4), (2,5), (2,6), (2,7), (2,8) \\
    \hline
    \end{tabular}
    \end{center}
    \caption{\footnotesize{A correspondence between topology name $U_{(n,i)}$ (size $n$ and index $i$) to an explicit list of edges. $D$ is short for edge-diameter.}}
    \label{table_trees_edges}
\end{table}
\section{Appendix: Code}
\label{appendix_section_code}
In this section we discuss our \emph{Sage} code for the sake of completeness. Section~\ref{subsection_brief_code} gives a few short comments regarding the code including where to find it. Section~\ref{subsection_outputs} summarizes outputs that appear throughout the text.

\subsection{About the Code}
\label{subsection_brief_code}
You may find the code on arXiv, ``hidden'' as additional source file named ``\emph{{STTLP-sage-python3-source.zip}}'' in the source-code of this paper (before it was compiled to the PDF file that you are reading). The archive contains a few scripts, and a few generated logs.

The code is written in the Sage language~\cite{SageWebsite}, whose syntax is mostly that of Python 3. In fact, some functionality of the code can be run purely in Python, such as analyzing topologies and enumerating their STTs. However, the ``core'' of solving LPs and manipulating polytopes does require Sage.

The code is divided into multiple scripts, based on a logical division that is explained in an opening comment within each file. The main script also has functions dedicated to reproducibility, to enable a relatively simple way to compute the various results discussed throughout the paper. While the code has been cleaned and made readable, this is by no means ``production level'' code.

In order to run the code, you need a standard installation of Sage (and, indirectly, of Python 3), and to make sure that you have the `networkx' python package installed (other imported packages such as `time' and `math' are standard in python). It is possible to install Sage locally, further details are on the web~\cite{SageInstallation}.
%Yaniv: Some notes there may be outdated. At least, this was the case when I installed it myself around November 2024.
For running short code, you may also use an online server such as~\cite{SageFreeServer}. ``Short'' mostly refers to time, each run is terminated after about a minute or two. There is also a limit to the length of the script, which is not reached with ``reasonable'' code, but can be reached if the script contains long hard-coded values.\footnote{There are $\approx 2000$ different primary directions which find non-integer vertices (not all unique), for each of the seven small topologies. Computing them is slow, so we pre-computed and hard-coded them and other demanding computations. This makes the code long, and as a side-effect we have these values documented.}

\subsection{Code Outputs (Tables Summary)}
\label{subsection_outputs}
We summarize the outputs of the script in several tables throughout the paper:
\begin{enumerate}
    \item Table~\ref{table_trees_summary_analysis_vertices_etc_FULL} (and Table~\ref{table_trees_summary_analysis_vertices_etc_FEW}, its subset): a table that summarizes our analysis over all the small topologies up to size $n \le 8$. One row per topology, with multiple details.
    
    \item Table~\ref{table_trees_edges} maps from size and index to the actual edges of the topology, to be concrete.

    \item Table~\ref{table_integrality_gaps} and Table~\ref{table_approximation_ratios}: List integrality gaps and approximation ratios of the rounding scheme in Definition~\ref{definition_rounding_scheme} with respect to each of the small topologies for which the gaps are larger than $1$ (topologies with non-integer vertices in $D$-space).

    \item Table~\ref{table_path_tree_fractions} analyzes sampled vertices of the polytope due to path topologies of different lengths, see further discussion in its caption, and in Section~\ref{subsection_with_and_without_z_fourier_motzkin}.

    \item Figure~\ref{figure_sage_3D_example}: Not computational per se, and not a table, yet it is still an output.
\end{enumerate}

\end{document}